\documentclass{article}
\usepackage{amsmath,amsfonts,graphicx,amsthm,verbatim}
\usepackage{hyperref}
\hypersetup{citecolor=blue,colorlinks=true}
\usepackage[authoryear,comma]{natbib}

\newtheorem{theorem}{Theorem}
\newtheorem{lemma}{Lemma}
\newtheorem{corollary}{Corollary}

\newtheorem{definition}{Definition}
\newtheorem{example}{Example}

 \numberwithin{equation}{section}

\usepackage{cleveref}
\usepackage{autonum}

\newcommand{\btheta}{{\boldsymbol{\theta}}}
\newcommand{\wtheta}{{\boldsymbol{\widetilde\btheta}}}
\newcommand{\X}{{\boldsymbol{X}}}
\newcommand{\Y}{{\boldsymbol{Y}}}
\newcommand{\Z}{{\boldsymbol{Z}}}

\newcommand{\be}{\begin{equation}}
\newcommand{\ee}{\end{equation}}
\newcommand{\ba}{\begin{eqnarray}}
\newcommand{\ea}{\end{eqnarray}}
\newcommand{\bee}{\begin{equation*}}
\newcommand{\eee}{\end{equation*}}
\newcommand{\baa}{\begin{eqnarray*}}
\newcommand{\eaa}{\end{eqnarray*}}
\newcommand{\nn}{\nonumber\\}
\newcommand{\bo}[1]{\boldsymbol{#1}}

\begin{document}

\iffalse
%%%%%%%%%%%%%%%%%%%%%%%%%%%%%%%%%%%%%%%%%%%%%%%%%%%%%%%%%%%%%%
\title{ An Optimal Test for the Additive Model with Discrete or Categorical Predictors
}
%\subtitle{An Optimal Test for the Additive Model}

\titlerunning{An Optimal Test for the Additive Model}        % if too long for running head

\author{Abhijit Mandal
}

\institute{A. Mandal \at
	Department of Mathematics, Wayne State University,
	Detroit, U.S.A. \\
	\email{abhijit.mandal@wayne.edu}          
}

\date{Received: date / Revised: date}
% The correct dates will be entered by the editor

%%%%%%%%%%%%%%%%%%%%%%%%%%%%%%%%%%%%%%%%%%%%%%%%%%%%%%%%%%%%%%

%%%%%%%%%%%%%%%%%%%%%%%%%%%%%%%%%%%%%%%%%%%%%%%%%%%%%%%%%%%%%%%%%%%%%%%%%%%%%%%%%%%%%%%%%%%%%%%%%%%%%%%%%%%%%%%%%%%%%%%%%%%%

\fi 

%\iffalse

\title{An Optimal Test for the Additive Model with Discrete or Categorical Predictors}
\author{Abhijit Mandal \\%}
% \address{
Department of Mathematics, Wayne State University\\
Detroit, MI 48202,
U.S.A.
% \email{abm92@pitt.edu}
}

%\maketitle

%\fi

%\begin{abstract}
%%%%%%%%%%%%%%%%%%%%%%%%%%%%%%%%%%%%%%%%%%%%%%%%%%%%%%%%%%%%%%%%%%%%%%%%%%%%%%%%%%%%%%%%%%%%%%%%%%%%%%%%%%%%%%%%%%%%%%%%%%%%

\date{}
\maketitle

\begin{abstract}
In multivariate nonparametric regression the additive models are very useful when a suitable parametric model is difficult to find. The backfitting algorithm is a powerful tool to estimate the additive components. However, due to complexity of the estimators, the asymptotic $p$-value of the associated test is difficult to calculate without a Monte Carlo simulation. Moreover, the conventional tests assume that the predictor variables are strictly continuous. In this paper, a new test is introduced for the additive components with discrete or categorical predictors, where the model may contain continuous covariates.  This method is also applied to the semiparametric regression to test the goodness-of-fit of the model. These tests are asymptotically optimal in terms of the rate of convergence, as they can detect a specific class of contiguous alternatives at a rate of $n^{-1/2}$.  An extensive simulation study is presented to support the theoretical results derived in this paper. Finally, the method is applied to a real data to model the diamond price based on its quality attributes and physical measurements.

\iffalse
\fi 
\end{abstract}



%%%%%%%%%%%%%%%%%%%%%%%%%%%%%%%%%%%%%%%%%%%%%%%%%%%%%%%%%%%%%%%%%%%%
%%%%%%%%%%%%%%%%%%%%%%%%%%%%%%%%%%%%%%%%%%%%%%%%%%%%%%%%%%%%%%%%%%%%
%
%\noindent
%{\textbf{AMS 2010 subject classification}}\textbf{:} Primary 62G10; Secondary 62J12, 62G20.
%
%\noindent{\textbf{Keywords}}:  Additive models; Backfitting algorithm; Generalized likelihood ratio test; Local polynomial regression, Semiparametric models.

% \begin{keywords}
% Additive models; Backfitting algorithm; Generalized likelihood ratio test; Local polynomial regression, Semiparametric models. 
% \end{keywords}

% \noindent
% {\textbf{AMS 2010 subject classification}}\textbf{:} Primary 62G10; Secondary 62J12, 62G20.

%%%%%%%%%%%%%%%%%%%%%%%%%%%%%%%%%%%%%%%%%%%%%%%%%%%%%%%%%%%%%%%%%%%%
%%%%%%%%%%%%%%%%%%%%%%%%%%%%%%%%%%%%%%%%%%%%%%%%%%%%%%%%%%%%%%%%%%%%

\section{Introduction}
The additive model is a widely used multivariate smoothing technique. It was originally suggested by \cite{MR650892} and popularized due to extensive discussion in \cite{MR1082147}.  It models a random sample $\{(Y_i, \X_i) : i=1,2, \cdots, n\}$ by
\be
Y_i = \alpha + \sum_{p=1}^P m_p(X_{pi}) + \epsilon_i,
\label{model00}
\ee
where the random error $\epsilon_i$ has mean zero, constant variance $\sigma^2$, and the additive component $m_p$ is an unknown smooth function for $p=1,2, \cdots, P$. \cite{MR790566,MR840516} have shown that the additive model reduces a full $P$-dimensional nonparametric regression effectively to a one-dimensional problem by fitting the model with the same asymptotic efficiency, i.e., an optimal convergence rate of $n^{-2/5}$ for twice continuously differentiable functions. So, it  has the very desirable property of reducing the ``curse of dimensionality''  in a satisfactory manner. In this paper, the additive function is estimated using the backfitting algorithm proposed by \cite{MR994249}. \cite{MR1429922, MR1701398} and \cite{MR1763322} studied asymptotic properties of the backfitting estimators.  In the literature, there are several other algorithms, such as the marginal integration estimation method  \citep{MR1310230}, the estimating equation method  \citep{MR1742496} 
% the Fourier series approximation approach of \cite{MR1925973}, the linear wavelet strategies of \cite{MR1864162}, the nonlinear wavelet estimation method of \cite{MR2063986}, using the block coordinate relaxation algorithm of \cite{MR1822091},
and the Bayesian backfitting algorithm  \citep{MR1820768}, among others.

To our knowledge, there are relatively limited theoretical results  on the testing problem for the additive models where discrete or categorical (possibly mixed with continuous) explanatory variables are considered. \cite{MR1891823} and \cite{MR1985882} considered marginal integration estimators to construct tests for testing the additive components with continuous variables. The asymptotic critical values of these tests are difficult to obtain due to the complicated expressions for the bias and the variance of the test
statistic. Moreover, the authors observed that the asymptotic accuracy of their result is  limited for small and moderate sample sizes. In the same setup, \cite{MR1833962} and \cite{MR2201017} proposed the generalized likelihood ratio (GLR) test for testing the significance of additive components using backfitting estimators. The idea is based on comparison of pseudo-likelihood functions under the null and the alternative hypotheses, which leads to the log-ratio of the variance estimators under the null and the alternative. Similar to the maximum likelihood ratio tests in parametric models, the GLR test has an important fundamental property that the asymptotic null distribution of the test is independent of nuisance parameters and functions. This property is referred to as the Wilks phenomenon.  The GLR test is asymptotically distribution-free; and it is asymptotically optimal in terms of convergence of the nonparametric hypothesis testing problem (see \citealp{MR1257978} and \citealp{MR1425962}). 
However, the authors mentioned that the GLR test may not be accurate as the test statistic contains an unknown bias term. So, a Monte Carlo simulation or a bootstrap technique is performed to calculate the $p$-value of the test. This somehow restricts the method for being widely applicable among the general practitioners. 

%It is a real problem in high-dimensional cases such as the genome-wide association study (GWAS) where there are around two million single nucleotide polymorphisms (SNPs) under the study; and when the SNPs are grouped into genes there are about 20,000 genes of different sizes. So, the simulation based $p$-values are  highly time consuming. Moreover, the above mentioned tests   assume that the predictor variables are strictly continuous, so it cannot be applied in GWAS where each SNP takes maximum three values (originally a categorical variable).   On the other hand, the test constructed in this paper is simple and computationally efficient, so it can be easily applied to the GWAS level. To focus on the theoretical development, the application of this method in GWAS is presented in a separate paper (see \citealp{mandal_GWAS}).

In this paper, we  propose a GLR test for the additive components having discrete or categorical valued predictors, while the model may contain continuous valued covariates. For categorical predictors this test may be regarded as the generalized analysis of covariance (ANCOVA), where covariates are modeled by nonparametric functions and the normality assumption on the error term is not required. In this case, the predictors may be referred as treatment or block of the design of experiment.

The rest of the paper is organized as follows. We gave an overview of the nonparametric additive model and the semiparametric additive model in Sections \ref{sec:nonpara} and \ref{sec:semiparametric}, respectively. In Section \ref{sec:GLR_main}, we introduced the GLR test and presented the  theoretical properties. An extensive simulation studies are performed in Section \ref{sec:simulation} to explore the behavior of the proposed test. In Section \ref{sec:real}, we have applied our method to analyze a real data containing diamond price, and we proposed an appropriate model based on quality attributes and physical measurements.  Section \ref{sec:concluding} has some concluding remarks.  The assumptions of the theorems are given in Appendix A.  Appendix B presents a brief description of the backfitting estimator. The  proofs of the theorems are given in Appendix C.

\section{The Nonparametric Additive Model} \label{sec:nonpara}
Let us consider a one dimensional response variable $Y$, a $P$-dimensional predictor $\X = (X_1, X_2, \cdots, X_P)^T$ and an additional $Q$-dimensional covariate $\Z = (Z_1, Z_2, \cdots, Z_Q)^T$. We assume that $\X$ contains only discrete or categorical variables, but $\Z$ may contain all type of variables -- categorical, discrete or continuous. If $X_p$ is a discrete valued random variable, then  $k_p$ denotes the number of  distinct values,  $x_{p1}, x_{p2}, \cdots, x_{pk_p}$, where $k_p<\infty$ for all $p=1, 2, \cdots, P$. If $X_p$ is a categorical variable, then $x_{p1}, x_{p2}, \cdots, x_{pk_p}$ are different levels of the variable.  Let $(Y_1, \X_1, \Z_1), \cdots, (Y_n, \X_n, \Z_n)$ be a random sample of size $n$ from $(Y, \X, \Z)$. The nonparametric additive model is given by
\be
Y_i = \alpha + \sum_{p=1}^P m_p(X_{pi}) + \sum_{q=1}^Q m_{P+q}(Z_{qi}) + \epsilon_i,
\label{model0}
\ee
where the random error $\epsilon_i$ has mean zero and a constant variance $\sigma^2$. To ensure identifiability of  components of the additive model, we set $E(m_p(X_{pi}))=0$ for all $p=1, 2, \cdots, P$, and $E(m_{P+q}(Z_{qi}))=0$ for all $q=1, 2, \cdots, Q$. The intercept parameter $\alpha = E(Y_i)$ is generally estimated by $\hat\alpha = \bar{Y} = \sum_i Y_i/n$. The backfitting estimator is used to estimate the nonparametric functions $m_d(\cdot)$ for $d=1,2,\cdots, P+Q$. For the ease of readability,  a discussion on the backfitting estimators is given in Appendix B. 
 We have divided regressors into two groups -- predictor and covariate, as we  construct a test of significance for predictors only, whereas their effect is  adjusted by  ``nuisance'' covariates. In other words, if we are interested to test the effect of a subset of regressors in the additive model, we name those regressors as predictors and remaining regressors as covariates. We require all predictors to be discrete or categorical variables, but there is no restriction on covariates.  %In the GWAS setup, when we are interested to test the genetic effect in a trait or disease, $\X$ is the set of SNPs in a gene, and $\Z$ is related environmental and demographic factors.

\section{The Semiparametric Additive Model} \label{sec:semiparametric}
The semiparametric additive model (SAM) is the combination of a parametric model and a nonparametric additive model. Here, some of the additive components are modeled parametrically while the remaining ones are unspecified and are estimated nonparametrically. First, we model predictors parametrically and covariates nonparametrically, then a generalized model is considered. In general, the predictors are assumed to be discrete valued random variables. However, if the predictors are ordinal categorical variables, their order or rank may also be modeled parametrically. Let us consider the following SAM model: %under the setup of the previous section:
\be
Y_i = \alpha + \sum_{p=1}^P m_{p,\btheta_p}(X_{pi}) + \sum_{q=1}^Q m_{P+q}(Z_{qi}) + \epsilon_i,
\label{sam}
\ee
where $\mathcal{M}_{p,\Theta_p} = \{m_{p,\btheta_p}(X_{pi}), \btheta_p \in \Theta_p\}$ is a family of parametric functions for $p=1,2, \cdots,P$. We assume that $m_{p,\btheta_p}(\cdot)$ is completely known except for the value of the parameter $\btheta_p$, $p=1,2, \cdots,P$. \cite{opsomer1999root} and \cite{MR2396026} have studied this model when the parametric models are linear functions. One might be interested to build a SAM when the main interest of study is to precisely quantify the effect of the predictors $X_1, X_2, \cdots, X_P$ on the dependent variable $Y$, but the relationship is observed in the presence of ``nuisance'' covariates $Z_1, Z_2, \cdots, Z_Q$. The use of the parametric forms for  predictors, if properly specified, allows us to make an easily interpretable inference about their effect on $Y$. On the other hand,  modeling  covariates nonparametrically, one may avoid  potential introduction of bias in the estimated relationship between  predictors and $Y$. Another possible situation when a SAM would be useful, if someone is fairly confident about the shape of the relationship between  predictors and $Y$, but not about that of the other covariates. It can be shown that by modeling some predictors using appropriate parametric functions, the risk of over-fitting the model is reduced by decreasing the overall degrees of freedom of the test. 

To ensure the identifiability of the model, we  assume that the expectation of the parametric term is zero, i.e. $E[m_{p,\btheta_p}(X_{p})] = 0$ for all $p=1,2, \cdots, P$, and $E(m_{P+q}(Z_{qi}))=0$ for all $q=1, 2, \cdots, Q$. We   consider the case where $m_{p,\btheta_p}(\cdot)$ is a polynomial for all $p=1,2, \cdots, P$.
As $X_p$ takes $k_p$  values, one needs at most $(k_p-1)$ parameters to completely specify $m_{p,\btheta_p}(\cdot)$. So, we assume that $m_{p,\btheta_p}(\cdot)$ is a polynomial of degree $r_p$, where  $0<r_p<k_p-1$, $p=1,2,\cdots,P$. Therefore, with a slight abuse of notation, we write
\be
m_{p,\btheta_p}(X_{p}) = \alpha_p + \sum_{s=1}^{r_p}\theta_{ps} X_{p}^s, 
\label{mtheta}
\ee
where $\btheta_p=( \theta_{p1}, \theta_{p2}, \cdots, \theta_{pr_p})^T$,  $p=1,2,\cdots,P$. Note that $\alpha_p$ is not an independent parameter, it just makes $m_{p,\btheta_p}$ centered at zero. Let us define $\btheta = (\alpha^*, \btheta_1^T, \btheta_2^T, \cdots, \btheta_P^T)^T$, where $\alpha^* = \sum_p \alpha_p$. 
For $a<b$, we  define
\be
{}_{a}^{b}\X_{(p)} = 
\begin{bmatrix}
	X_{p1}^a & X_{p1}^{a+1} & \dots  & X_{p1}^b \\
	X_{p2}^a & X_{p2}^{a+1} & \dots  & X_{p2}^b \\
	\vdots & \vdots  & \ddots & \vdots \\
	X_{pn}^a & X_{pn}^{a+1} & \dots  & X_{pn}^b \\
\end{bmatrix},
\label{xp}
\ee
and if $a=b$ then ${}_{a}^{b}\X_{(p)}$ is a vector containing the first column of Equation (\ref{xp}). Let $\X^*=(\bo{1}_n, {}_{1}^{r_1}\X_{(1)}, $ ${}_{1}^{r_2}\X_{(2)}, \cdots,{}_{1}^{r_P}\X_{(P)})$. 
Then, following \cite{MR970977}, the estimates of the additive components are derived from the backfitting algorithm as the solution of the following equations: 
\ba
\wtheta &=&  \left(\X^{*T}  \X^* \right)^{-1} \X^{*T} \left(\Y^*  - \sum_{q=1}^Q \widetilde{\bo{m}}_{P+d} \right), \nn
\widetilde{\bo{m}}_{P+q} &=& \bo{S}_{P+q}^* \left(\Y^* -  \X^* \wtheta - \sum_{d\neq q} \widetilde{\bo{m}}_{P+d} \right), \mbox{ for } q=1,2,\cdots,Q,
\ea
provided $(\X^{*T}  \X^* )^{-1}$ exists. Here, $\bo{S}_{P+q}^*$  is the centered smoothing matrix $\bo{S}_d^*$ for $d=P+q$ as defined after Equation (\ref{normal_eqn}) in Appendix B. Suppose $\bo{W}_{P+q}$ is the additive smoother matrix of $\bo{m}_{P+q}$, so that the backfitting estimate of $\bo{m}_{P+q}$ is $\widetilde{\bo{m}}_{P+q} = \bo{W}_{P+q} (\Y^* -  \X^* \wtheta)$ for $q=1,2,\cdots,Q$. Let us define $\bo{W}_{[Z]}= \sum_{q=1}^Q \bo{W}_{P+q}$ and $\widetilde{\bo{m}}_{[Z]}=\sum_{q=1}^Q \widetilde{\bo{m}}_{P+q}$. Then, the above normal equations are solved non-iteratively as
\ba
\wtheta &=&  \left(\X^{*T} \left(\bo{I}_n - \bo{W}_{[Z]} \right) \X^* \right)^{-1} \X^{*T} \left(\bo{I}_n - \bo{W}_{[Z]}\right) \Y^* ,
\nn
\widetilde{\bo{m}}_{[Z]} &=& \bo{W}_{[Z]} \left(\Y^* - \X^* \wtheta  \right).
\ea

Sometimes the experimenter may have a prior knowledge about some of the variables and would like to model them parametrically, whereas keeping other variables in the nonparametric model. In that case, one may consider the following generalized SAM model:
% \begin{multline}
% Y_i = \alpha + \sum_{p=1}^{P_1} m_{p,\btheta_p}(X_{pi}) + \sum_{p=P_1+1}^P m_p(X_{pi}) \\ + \sum_{q=1}^{Q_1} m_{P+q, ,\btheta_{P+q}}(Z_{qi})  + \sum_{q=Q_1+1}^Q m_{P+q}(Z_{qi}) + \epsilon_i,
% \label{mixed}
% \end{multline}
\be
Y_i = \alpha + \sum_{p=1}^{P_1} m_{p,\btheta_p}(X_{pi}) + \sum_{p=P_1+1}^P m_p(X_{pi})  + \sum_{q=1}^{Q_1} m_{P+q, \btheta_{P+q}}(Z_{qi})  + \sum_{q=Q_1+1}^Q m_{P+q}(Z_{qi}) + \epsilon_i,
\label{mixed}
\ee
where $P_1\leq P,\ Q_1 \leq Q$ and $\mathcal{M}_{p,\Theta_p} = \{m_{p,\btheta_p}(X_{pi}), \btheta_p \in \Theta_p\}$ for $p=1,2, \cdots,P_1$,  $\mathcal{M}_{P+q,\Theta_{P+q}} = \{m_{P+q,\btheta_{P+q}}(Z_{qi}), \btheta_{P+q} \in \Theta_{P+q}\}$ for $q=1,2, \cdots,Q_1$ are families of parametric functions. For simplicity, we assume that $m_d$ is a polynomial of degree $r_d$ for all $d=1,2,\cdots,P_1$ and $d=P+1,P+2,\cdots,P+Q_1$.

\section{The Generalized Likelihood Ratio Test} \label{sec:GLR_main}
Let us consider the generalized SAM model defined in Equation (\ref{mixed}). For this model, one may be interested mainly in two type of tests based on the predictor variables -- a goodness-of-fit test for the parametric function and a model utility test for the nonparametric function. First, we present the generalized test that includes both type of tests, then we discuss about the individual test. So, we are now interested in the following null hypothesis:
% \ba
% H_0^{**} &:& m_p (\cdot) \in \mathcal{M}_{p,\Theta_p} \mbox{ for } p=1,2,\cdots, P_1,\nn
% && \mbox{and } m_p (\cdot) =0 \mbox{ for } p=P_1+1,P_1+2,\cdots, P.
% \label{null3}
% \ea
\be
H_0 : m_p (\cdot) \in \mathcal{M}_{p,\Theta_p} \mbox{ for } p=1,2,\cdots, P_1,
\mbox{ and } m_p (\cdot) =0 \mbox{ for } p=P_1+1,P_1+2,\cdots, P.
\label{null3}
\ee
 Under $H_0$, we define $\widetilde m_d(\cdot)$ and $\widetilde m_{d,\widetilde \btheta_d}$ as the backfitting estimators of $m_d(\cdot)$ and  $m_{d, \btheta_d}$, respectively. 
Then, the residual sum of squares, under $H_0$, is given by
\be
RSS_0 =\sum_{i=1}^n \left(Y_i - \hat\alpha - \sum_{p=1}^{P_1} \widetilde m_{p,\widetilde \btheta_p}(X_{pi})  - \sum_{q=1}^{Q_1} \widetilde m_{P+q, \widetilde \btheta_{P+q}}(Z_{qi})  - \sum_{q=Q_1+1}^Q \widetilde m_{P+q}(Z_{qi}) \right)^2.
\ee
As we are testing only for  predictors by keeping covariates unchanged, the unconstrained model is
\be
Y_i = \alpha + \sum_{p=1}^P m_p(X_{pi})  + \sum_{q=1}^{Q_1} m_{P+q, \btheta_{P+q}}(Z_{qi})  + \sum_{q=Q_1+1}^Q m_{P+q}(Z_{qi}) + \epsilon_i.
\label{mixed123}
\ee
Under this model, the residual sum of square  is given by
\be
RSS_1 =\sum_{i=1}^n \left(Y_i - \hat\alpha - \sum_{p=1}^P \widehat m_p(X_{pi}) - \sum_{q=1}^{Q_1} \widehat m_{P+q, \widehat\btheta_{P+q}}(Z_{qi})  + \sum_{q=Q_1+1}^Q \widehat m_{P+q}(Z_{qi})\right)^2,
\label{rssh0_new}
\ee
where $\widehat m_d(\cdot)$  and $\widehat m_{d,\widehat \btheta_d}$ are the backfitting estimators of $m_d(\cdot)$ and $ m_{d, \btheta_d}$, respectively, under the unconstrained model.
The  generalized likelihood ratio (GLR) test statistic for testing  null hypothesis $H_0$ is defined as
\be
\lambda_n (H_0) = %n \log \left(\frac{RSS_0}{RSS_1}\right) \approx
\frac{n (RSS_0 - RSS_1)}{RSS_1}.
\label{wilks_stat_z_new}
\ee
If the difference between $RSS_0$ and $RSS_1$ is small, then the GLR test statistic may be approximated by $n \log (\frac{RSS_0}{RSS_1})$. In the parametric model, this is equivalent to the log-likelihood ratio test statistic, where estimators are replaced by the corresponding maximum likelihood estimators.  Generally, the nonparametric maximum likelihood estimate does not exist and even when it does exist, the resulting maximum likelihood ratio test is not optimal (see \citealp{MR2201017, MR970470}).  So, the GLR test statistic may be regarded as a log-ratio of the quasi-likelihoods.

We assumed homoscedasticity of the error term, i.e. error $\epsilon_i$ in model (\ref{mixed}) has a constant variance. If this assumption is not valid, one may consider a GLR statistic by taking weighted residual sum of squares. The subsequent analysis and the backfitting algorithm will be modified similar to the weighted likelihood approach for the parametric models. 

\subsection{Asymptotic Distribution} \label{sec:asymp_general}
To derive the null distribution of the GLR  test statistic, let us define $c_p = (\sqrt{c_{p1}}, \sqrt{c_{p2}}, \cdots, \sqrt{c_{pk_p}})^T$, where  $c_{pj}=P(X_p = x_{pj})$ for $j=1, 2, \cdots, k_p$ and $p=1, 2, \cdots, P$. We also define ${}_{a}^{b}\Z_{(P+q)}$ in a similar way of ${}_{a}^{b}\X_{(p)}$  as defined in Equation (\ref{xp}) by replacing $\X$ with $\Z$. Let us denote $\bo{T}^*=(\bo{1}_n, {}_{1}^{r_1}\X_{(1)}, {}_{1}^{r_2}\X_{(2)}, \cdots, $ $ {}_{1}^{r_{P_1}}\X_{(P_1)}, {}_{1}^{r_{P+1}}\bo{Z}_{(P+1)}, {}_{1}^{r_{P+2}}\bo{Z}_{(P+2)}, $ $ \cdots , {}_{1}^{r_{P+Q1}}\bo{Z}_{(P+Q_1)})$.
% We further define
% \be
% \bo{S}^* = \sum_{p=1}^{P_1} \left(\bo{S}_p 
% -   \X_{(p)}^* \left(\X_{(p)}^{*T}  \X_{(p)}^*\right)^{-1} \X_{(p)}^{*T} \right)  + \sum_{p=P_1+1}^{P} \bo{S}_p^*.
% \label{s_semi}
% \ee
% Suppose $\bo{S}^*$ can be diagonalized as $\bo{S}^* = \bo{R} \bo{\Lambda} \bo{R}'$, where $\bo{\Lambda} = {\rm{diag}}(\xi_1, \xi_2, \cdots, \xi_k)$ is the diagonal matrix of non-zero eigenvalues of $\bo{S}^*$,  $\bo{R}$ is the matrix of the corresponding orthonormal eigenvectors, and $k$ is the rank of $\bo{S}^*$. 
Suppose $I = \sum_{p=1}^{P_1} (k_p - r_p - 1) + \sum_{p=P_1+1}^P k_p$, and $\bo{\Sigma}_1$ is a $I\times I$ dimensional block diagonal matrix whose $p$-th diagonal block is an identity matrix of order $(k_p - r_p - 1)$ if $p\leq P_1$, and $\left(\bo{I}_{k_p} - c_p c_p^T \right)$  if $p>P_1$. Define another $I\times I$ dimensional  block  matrix $\bo{\Sigma}_2$, whose $p$-th diagonal block is the identity matrix of order $(k_p - r_p - 1)$ if $p\leq P_1$, and of order $k_p$ if $p>P_1$. For $p\neq p' \in \{1, 2, \cdots, P\}$, the $pp'$-th off-diagonal block of $\bo{\Sigma}_2$ is given by
\be
\Sigma_{pp'} = 
\lim_{n \rightarrow \infty}  \left(\bo{R}_p^T  \bo{R}_p\right)^{-1/2} \bo{R}_p^T  \bo{R}_{p'} \left(\bo{R}_{p'}^T  \bo{R}_{p'}\right)^{-1/2},
\label{corjjsem}
\ee
 where $\bo{R}_p = {}_{r_p+1}^{k_p-1}\X_{(p)}$ if $p\leq P_1$ and $\bo{R}_p = {}_0^{k_p-1}\X_{(p)}$ if $p > P_1$.
  Then, the following theorem gives the asymptotic null distribution of the GLR  statistic. 

%the $ij$-th element of the $pp'$-th off-diagonal block is the $ij$-th element of $\bo{\Sigma}_{pp'}$ if $p\leq P_1$, and $\sigma_{pp',ij}$ as given in (\ref{corjj})  if $p>P_1$.

\begin{theorem}
	Suppose that regularity conditions (C1)--(C8) in Appendix A hold. Further assume that the limit of $n^{-1} \bo{T}^{*T} \bo{T}^*$ exists and it is invertible. Let us consider the unconstrained model (\ref{mixed123}) and the null hypothesis $H_0$ in (\ref{null3}), where  $m_{p,\btheta_p}$ is a polynomial of degree $r_p$ and $0<r_p<(k_p-1)$, for $p =1, 2, \cdots, P_1$. Then, under $H_0$, the asymptotic distribution of the GLR  test statistic coincides with $\sum_{i=1}^s \lambda_i V_i^2$, where $V_1, V_2, \cdots, V_s$ are independent and identically distributed (i.i.d.) standard normal variables, $\lambda_1, \lambda_2, \cdots, \lambda_s$ are non-zero eigenvalues of $\bo{\Sigma}_1\bo{\Sigma}_2 \bo{\Sigma}_1$ and $s$ is the rank of  $\bo{\Sigma}_1\bo{\Sigma}_2 \bo{\Sigma}_1$.
	\label{theorem:chi3}
\end{theorem}
% 
% Theorem \ref{theorem:chi3} generalizes Theorems \ref{theorem:chi} and \ref{theorem:chi2}. In fact Theorem \ref{theorem:chi} is a special case of Theorem \ref{theorem:chi2} where $r_p = k_p $ for all $p=1,2, \cdots, P$. In this case $k_p$ parameters completely specifies the distribution of $X_p$. The proof of Theorem \ref{theorem:chi3} can derived in the same way as done for Theorem \ref{theorem:chi}.

 The proof of the theorem is given in Appendix B. Theorem \ref{theorem:chi3} shows that the asymptotic null distribution of the GLR test statistic is a linear combination of chi-square variables. The critical region of the test may be calculated using the algorithm proposed by \cite{davies1980algorithm}.
Note that the null distribution does not depend on the modeling of covariates as we keep covariates unchanged under both the null and the alternative hypotheses. But, in practice, it reduces the possible over-fitting; thus the finite sample performance of the test improves due to parametric modeling those covariates. However, one must be careful while modeling  covariates, and it is important to verify whether such a parametric model is valid or not. A wrong model may cause severe power loss as demonstrated in the simulation studies. A special case of Theorem \ref{theorem:chi3} when the predictor variables are pairwise independent, the null distribution is reduced to a single chi-square as mentioned in the following corollary.

\begin{corollary}
	Suppose that the predictor variables are pairwise independent, and the assumptions of Theorem \ref{theorem:chi3} hold.  Then, under $H_0$, the asymptotic distribution of the GLR  test statistic is  a chi-square distribution with degrees of freedom $\sum_{p=1}^{P_1} (k_p - r_p - 1) + \sum_{p=P_1+1}^P (k_p-1)$. 
	\label{corollary:chi3}
\end{corollary}

It is shown in the simulation section that even if the predictors are not pairwise independent one may approximate the null distribution using Corollary \ref{corollary:chi3} unless predictors are strongly correlated. Simulation studies show that the approximation makes the test  little anti-conservative in small or moderate sample sizes. However, as sample size increases it gives a good approximation.

It is interesting to note that the asymptotic null distribution of the GLR test statistic does not depend on the nuisance parameters -- the design densities of $\Z$,  $m_d$ functions for the covariates and the error distribution. But, in general, it depends on the design densities of $\X$ as shown in Theorem \ref{theorem:chi3}. So, the Wilks phenomenon does not hold in the true sense, however, it holds good if predictors are pairwise independent.

The main advantage of our method is that it is easy to calculate the $p$-values of the test. In fact, the discrete valued predictors make the test simple. On the other hand, as shown by \cite{MR2201017} and \cite{MR2396026}, if the predictors are continuous the GLR test becomes complicated, and the asymptotic null distribution depends on  kernel  density functions and  bandwidth parameters. Moreover, the authors mentioned that the null distribution may not be accurate as the test statistic contains an unknown bias term. So, the null distribution is calculated by Monte Carlo simulation or using a  conditional bootstrap method. In our test, we do not need any additional conditions to choose the bandwidth parameter for continuous covariates as long as assumption (C5) holds. Therefore, for simplicity, we may use the same bandwidth parameter which is optimal for  estimation (see Section \ref{sec:simulation}). %However, the optimal bandwidth for the testing problem by maximizing the power function may be different, but it is beyond the scope of the present paper. 

The main contribution of this paper is that the GLR test is constructed for the additive components with discrete or categorical predictors.
The asymptotic distribution of the GLR test statistic is derived in \cite{MR2201017} with the strict assumption that all the predictors and covariates are continuous, more specifically, their marginal
distributions must be Lipschitz continuous on some bounded support. Keeping in mind the real applications, this restriction is  modified for the discrete or categorical predictors. For this reason, the proof of the theorem does not directly follow from the previous method, and we used a novel approach to derive the asymptotic distribution.

%\iffalse
\subsection{Power Function}
We now consider the power function of the GLR test. Let us take the following contiguous alternative hypothesis:
\be
H_1 : m_p (\cdot) = n^{-1/2} m_p^* (\cdot) \mbox{ for } p=1,2,\cdots, P,
\label{alt_new}
\ee
where $m_p^*$ is an additive function under the alternative hypothesis such that $E(m_p^*)=0$.
%Now, we  calculate the power function for the GLR semiparametric test. Let us consider the alternative hypothesis $H_1$ as mentioned in (\ref{alt_new}).
 We define $m_p'$ as the best fitted polynomial of degree $r_p$ to $m_p^*$ for $p = 1, 2, \cdots, P_1$, and $m_p' = m_p^*$ for $p = P_1 + 1, P_1 +  2, \cdots, P$. 
%Define $\bo{m}' = \sum_{p=1}^P \bo{m}_p'$, where $\bo{m}_p' = (m_p'(X_{p1}), \cdots, m_p'(X_{pn}))^T$.
% $\bo{\Lambda}$ and $\bo{R}$ are, respectively, the diagonal matrix of non-zero eigenvalues and the corresponding matrix of eigenvectors of $\bo{S}^*$ as defined in (\ref{s_semi}). 
%  Suppose $\bo{R}_{\cdot i}$ is the $i$-th column of $\bo{R}$ for $i=1,2,\cdots, k$, where $k$ is the rank of $\bo{S}^*$. 
Using the following theorem, we get the power function of the GLR  test.

\begin{theorem}
	Let us consider the notations and assumptions of Theorem \ref{theorem:chi3}. Then, under $H_1$, the asymptotic distribution of the GLR  test statistic coincides with $\delta^2 + \sum_{i=1}^s \lambda_i V_i^2$, where  $\delta^2 = \sum_{r,s=1}^{P} E(m'_r  m'_s)$.
	%  
	%  Suppose that regularity conditions (C1)--(C8) in Appendix hold, and the limit of $n^{-1} \bo{U}^{*T} \bo{U}^*$ exists and it is invertible. Then, the distribution of the GLR semiparametric test statistic, under $H_1$, coincides with $\sum_{i=1}^k \xi_i (U_i + \bo{R}_{\cdot i}^T \bo{m}')^2$, where $U_1, U_2, \cdots, U_k$ are independent standard normal variables.
	\label{theorem:power2}
\end{theorem}

It is interesting to note that the GLR test detects a specific class of contiguous alternatives at a rate of $n^{-1/2}$. So, the power of the test is asymptotically optimal in terms of the rate of convergence. For a fixed alternative the power of the GLR test convergences to one, i.e. the test is consistent. In case of pairwise independent predictors, the theorem simplifies as below. 

%(***)  
%Theorem \ref{theorem:power2} proves that, like the full nonparametric GLR test, the GLR semiparametric test is also asymptotically optimal in terms of its rate of convergence. Corollary \ref{corollary:power2} gives a special case of Theorem \ref{theorem:power2} when predictors are pairwise independent.  (***)

\begin{corollary}
	Suppose that the predictor variables are pairwise independent, and the assumptions of Theorem \ref{theorem:chi3} hold. Then, under $H_1$, the asymptotic distribution of the GLR  test statistic is   non-central chi-square  with the non-centrality parameter $ \sum_{p=1}^{P} E({m_p'}^2)$ and the degrees of freedom $\sum_{p=1}^{P_1} (k_p - r_p - 1) + \sum_{p=P_1+1}^P (k_p-1)$.
	\label{corollary:power2}
\end{corollary}

%\fi
 It is not surprising that the GLR statistic is $\sqrt{n}$-consistent, whereas most of the nonparametric tests have relatively slower rate.  Estimating the $m_p$ function corresponding to a discrete or categorical predictor $X_p$ is a finite dimensional problem as we assume that $k_p$, the domain of $X_p$, is  finite. Thus, all $m_p$ functions for $p=1,2,\cdots,P$ of the additive models in (\ref{mixed}) and (\ref{mixed123}) are equivalent to the parametric part of the semiparametric model. So, the corresponding convergence rate is consistent with Corollary 1 of \cite{opsomer1999root}. The same result is also obtained by \cite{MR970977}. However, this rate depends on the bandwidth parameter selected for smoothing the continuous covariates (if any). The  bandwidth parameter must be selected based on assumption (C5)  given in Appendix A. Even if when the model contains continuous covariates, as the null hypothesis is the significance only for the predictors, the GLR test  has a convergence rate similar to a parametric model. Intuitively, in large sample sizes, the effect of smoothers for continuous covariates cancels out when two residual sum of squares are subtracted  in the numerator of the GLR statistic in Equation (\ref{wilks_stat_z_new}). Meanwhile, the denominator (decided by $n$) of the GLR statistic converges to $\sigma^2$, the error variance of the additive model.  So, those continuous smooth functions do not have any major role in the rate of convergence of the GLR test as long as their bandwidth parameters are optimally selected.

We have started with a very general test in $H_0$ that includes both the goodness-of-fit test  and the model utility test. For this reason, the construction of $\bo{\Sigma_1}$ and $\bo{\Sigma_2}$ used in Theorem \ref{theorem:chi3} looks slightly complicated. However, if we are interested only in one type of test, those matrices becomes very simple. In these cases, the residual sum of squares $RSS_0$ and $RSS_1$  also have simpler expressions. Now, we discuss about these special cases. To make the procedure further simple, we assume that all covariates are modeled nonparametrically.

\subsection{The Goodness-of-Fit Test for the Semiparametric Model}
Let us consider the full nonparametric additive model given in (\ref{model0}). With respect that base model, we  now construct a   goodness-of-fit test for the semiparametric model (\ref{sam}), where all predictors are modeled parametrically. So, the null hypothesis for this problem is 
\be
H_0^* : m_p (\cdot) \in \mathcal{M}_{p,\Theta_p}, \mbox{ for } p=1,2,\cdots, P.
\label{null2}
\ee
%As the alternative hypothesis is the full nonparametric additive model (\ref{model0}),  the residual sum of squares for the full model  remain unchanged as given in (\ref{rssh0}). Now, t
The residual sum of squares, under $H_0^*$, is given by
\be
RSS_0^* =\sum_{i=1}^n \left(Y_i - \hat\alpha - \sum_{p=1}^P \widetilde m_{p,\wtheta_p}(X_{pi}) - \sum_{q=1}^Q \widetilde m_{P+q}(Z_{qi})\right)^2,
\ee
where $\widetilde m_{p,\wtheta_p}$ is the backfitting estimator of $m_{p,\btheta_p}$ under $H_0^*$ for $p=1,2,\cdots,P$, and $\widetilde m_{P +q}$ is the backfitting estimator of $m_{P+q}$ under $H_0^*$ for $q=1,2,\cdots,Q$. 
Under the unconstrained nonparametric additive model, the residual sum of square  is given by
\be
RSS_1^* =\sum_{i=1}^n \left(Y_i - \hat\alpha - \sum_{p=1}^P \widehat m_p(X_{pi}) - \sum_{q=1}^Q \widehat m_{P+q}(Z_{qi})\right)^2,
\label{rssh0}
\ee
where $\widehat m_1(\cdot), \widehat m_2(\cdot), \cdots , \widehat m_{P+Q}(\cdot)$ are the backfitting estimators under the full model given in (\ref{model0}). 
% In order to derive the null distribution, we define
% \be
% \bo{S}^* = \sum_{p=1}^P \left(\bo{S}_p 
% -   \X_{(p)}^* \left(\X_{(p)}^{*T}  \X_{(p)}^*\right)^{-1} \X_{(p)}^{*T} \right) .
% \ee
% Now $\bo{S}^*$ can be diagonalized as $\bo{S}^* = \bo{R} \bo{\Lambda} \bo{R}'$, where $\bo{\Lambda} = {\rm{diag}}(\xi_1, \xi_2, \cdots, \xi_k)$ is the diagonal matrix of non-zero eigenvalues of $\bo{S}^*$,  $\bo{R}$ is the matrix of the corresponding orthonormal eigenvectors, and $k$ is the rank of $\bo{S}^*$.
% % Suppose $\bo{R}_{\cdot i}$ is the $i$-th column of $\bo{R}$ for $i=1,2,\cdots, k$, and $\bo{m}^* = \sum_{p=1}^P \bo{m}_p^*$. 
Suppose $L = \sum_{p=1}^P (k_p - r_p- 1) $, and $\bo{\Sigma}_2$ is a $L\times L$ dimensional block matrix, whose $p$-th diagonal block is an identity matrix of order $(k_p - r_p- 1)$, and for  $p\neq p' \in \{1, 2, \cdots, P\}$ the $pp'$-th off-diagonal block of $\bo{\Sigma}_2$ is given in Equation (\ref{corjjsem}) 
%\be
%\Sigma_{pp'} = 
%\lim_{n \rightarrow \infty}  \left(\bo{R}_p^T  \bo{R}_p\right)^{-1/2} \bo{R}_p^T  \bo{R}_{p'} \left(\bo{R}_{p'}^T  \bo{R}_{p'}\right)^{-1/2},
%\label{corjjsem}
%\ee
 with $\bo{R}_p = {}_{r_p+1}^{k_p-1}\X_{(p)}$. Notice that $\bo{\Sigma}_1$,  defined in Section \ref{sec:asymp_general}, becomes an identity matrix in this setup  as $P_1=P$. Moreover, we need existence of  $n^{-1} \X^{*T} \X^*$ instead of $n^{-1} \bo{T}^{*T} \bo{T}^*$, where $\X^*$ is defined in Section \ref{sec:semiparametric} after Equation (\ref{xp}). 
Then, the following result gives the asymptotic null distribution of the GLR goodness-of-fit test statistic $\lambda_n (H_0^*) = 
\frac{n (RSS_0^* - RSS_1^*)}{RSS_1^*}$.

\begin{corollary}
	Suppose that the limit of $n^{-1} \X^{*T} \X^*$ exists and it is invertible, and regularity conditions (C1)--(C8) in Appendix A hold.  Let us consider the unconstrained model (\ref{model0}) and the null hypothesis $H_0^*$ in (\ref{null2}), where  $m_{p,\btheta_p}$ is a polynomial of degree $r_p$ and $0<r_p<(k_p-1)$ for all $p=1,2,\cdots, P$. Then,  the asymptotic distribution of the GLR goodness-of-fit test statistic, under $H_0^*$, coincides with $\sum_{i=1}^s \lambda_i V_i^2$, where $V_1, V_2, \cdots, V_s$ are i.i.d. standard normal variables, $\lambda_1, \lambda_2, \cdots, \lambda_s$ are non-zero eigenvalues of $\bo{\Sigma}_2$ and $s$ is the rank of  $\bo{\Sigma}_2$.
	\label{theorem:chi2}
\end{corollary}

In general, the asymptotic null distribution of the GLR goodness-of-fit test statistic is  a linear combination of chi-square variables. But, the distribution comes out to be a single chi-square if the predictor variables are pairwise independent. In this case, the degrees of freedom of the test statistic becomes $\sum_{p=1}^P (k_p - r_p - 1) $.  %The result is presented in the following corollary.

\iffalse
\begin{corollary}
	Suppose that the predictor variables are pairwise independent, and regularity conditions (C1)--(C8) in Appendix A hold. Further assume that the limit of $n^{-1} \X^{*T} \X^*$ exists and it is invertible.  Then, under $H_0^*$, the asymptotic distribution of the GLR semiparametric test statistic is a chi-square distribution with degrees of freedom $\sum_{p=1}^P (k_p - r_p - 1) $. 
	\label{corollary:chi2}
\end{corollary}
Corollary \ref{corollary:chi2} shows that the number of degrees of freedom of the GLR semiparametric test statistic is reduced due to parametric modeling. %However, one must be careful while modeling, and it is important to verify whether such a model is valid or not. Wrong model may cause severe power loss as demonstrated in the simulation section.

\fi

\subsection{Model Utility Test for the Nonparametric Additive Model} \label{sec:GLR}
Let us consider the null hypothesis that there is no association between $Y$ and $X_1, \cdots , X_p$, where $Z_1, \cdots , Z_q$ are  covariates of the nonparametric additive model  (\ref{model0}). The null hypothesis can be written as
\be
H_0^{**} : m_p (\cdot) = 0 \mbox{ for all } p=1,2,\cdots, P.
\label{null}
\ee
Let $\widetilde m_{P+1}(\cdot), \widetilde m_{P+2}(\cdot), \cdots , \widetilde m_{P+Q}(\cdot)$ be the backfitting estimators under $H_0^{**}$.
Then, the residual sum of squares, under $H_0^{**}$, is given by
\be
RSS_0^{**} =\sum_{i=1}^n \left(Y_i - \hat\alpha - \sum_{q=1}^Q \widetilde m_{P+q}(Z_{qi})\right)^2.
\ee
Under the unconstrained nonparametric additive model  (\ref{model0}), the residual sum of square  is given in Equation (\ref{rssh0}). So, the GLR nonparametric test statistic becomes $\lambda_n (H_0^{**}) = 
\frac{n (RSS_0^{**} - RSS_1^*)}{RSS_1^*}$. 
Suppose $\bo{\Sigma}_1$ is a $K\times K$ dimensional block diagonal matrix whose $p$-th diagonal block is $\left(\bo{I}_{k_p} - c_p c_p^T \right)$  for  $p=1,2,\cdots,P$. Define another $K\times K$ dimensional  block  matrix $\bo{\Sigma}_2$, whose $p$-th diagonal block is an identity matrix of order $k_p$, and for  $p\neq p' \in \{1, 2, \cdots, P\}$ the $ij$-th element of the $pp'$-th off-diagonal block of $\bo{\Sigma}_2$ is given by
\be
 \sigma_{pp',ij} = \frac{1}{\sqrt{c_{pi} c_{p'j}}} P( X_p = x_{pi}, X_{p'} = x_{p'j} ),
 \label{corjj_temp}
\ee
where $i=1,2,\cdots, k_p$ and $j=1,2,\cdots, k_{p'}$.
The following result gives the asymptotic distribution of the GLR nonparametric test statistic for testing the null hypothesis $H_0^{**}$.
\begin{corollary}
 Let us  assume that regularity conditions (C1)--(C8) in Appendix A hold.  Let us consider the unconstrained model (\ref{model0}) and the null hypothesis $H_0^{**}$ in (\ref{null}). Then, under $H_0^{**}$, the asymptotic distribution of the GLR nonparametric test statistic coincides with $\sum_{i=1}^s \lambda_i V_i^2$, where $V_1, V_2, \cdots, V_s$ are i.i.d. standard normal variables, $\lambda_1, \lambda_2, \cdots, \lambda_s$ are non-zero eigenvalues of $\bo{\Sigma}_1\bo{\Sigma}_2 \bo{\Sigma}_1$ and $s$ is the rank of  $\bo{\Sigma}_1\bo{\Sigma}_2 \bo{\Sigma}_1$.
 \label{theorem:chi}
\end{corollary}

If the  predictor variables are pairwise independent, the GLR nonparametric test statistic follows the chi-square distribution with degrees of freedom $\sum_{p=1}^P (k_p -1)$ under the null hypothesis. The setup of the GLR nonparametric test is similar to the classical ANCOVA when   predictors are categorical variables. In ANCOVA predictors are called treatments or blocks, and covariates are modeled parametrically. The goal is to test the treatment or block effect in the design of experiment. So, the GLR test is  generalization the classical ANCOVA, where covariates are modeled nonparametrically. Moreover, we do not need to assume that the error distribution is normal.

\section{Simulation} \label{sec:simulation}
In the first part of the simulation, we  check the null distribution of the GLR nonparametric test statistic for the hypothesis given in (\ref{null}). 
% It will be shown that the null distribution does not depend on the error distributions as long as it has a constant variance and its fourth moment is finite. 
Then, we demonstrate the power of the GLR test and compare it with the F-test associated with the nested linear models in the regression analysis. And finally, the performance of the  general GLR test for testing $H_0$ in (\ref{null3}) under the semiparametric model is presented. All numerical examples in this paper are performed using R software. The R code for the GLR test  will be provided on request.

\begin{figure}
%\vspace{-.5cm}
 \centering%
 \begin{tabular}{cc}
 (a) & (b) \\
 \includegraphics[height=7cm, width=7.5cm]{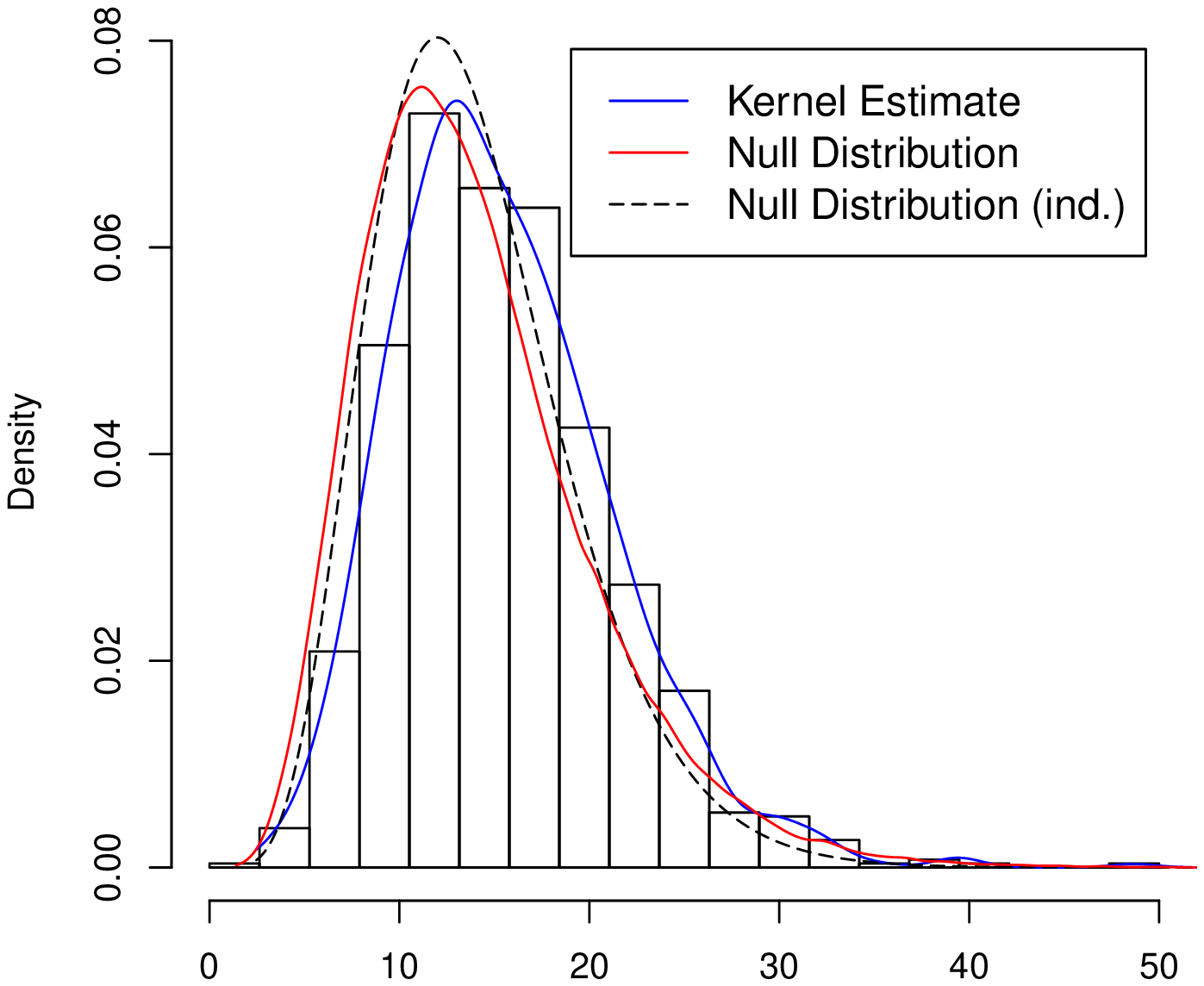} &
  \includegraphics[height=7cm, width=7.5cm]{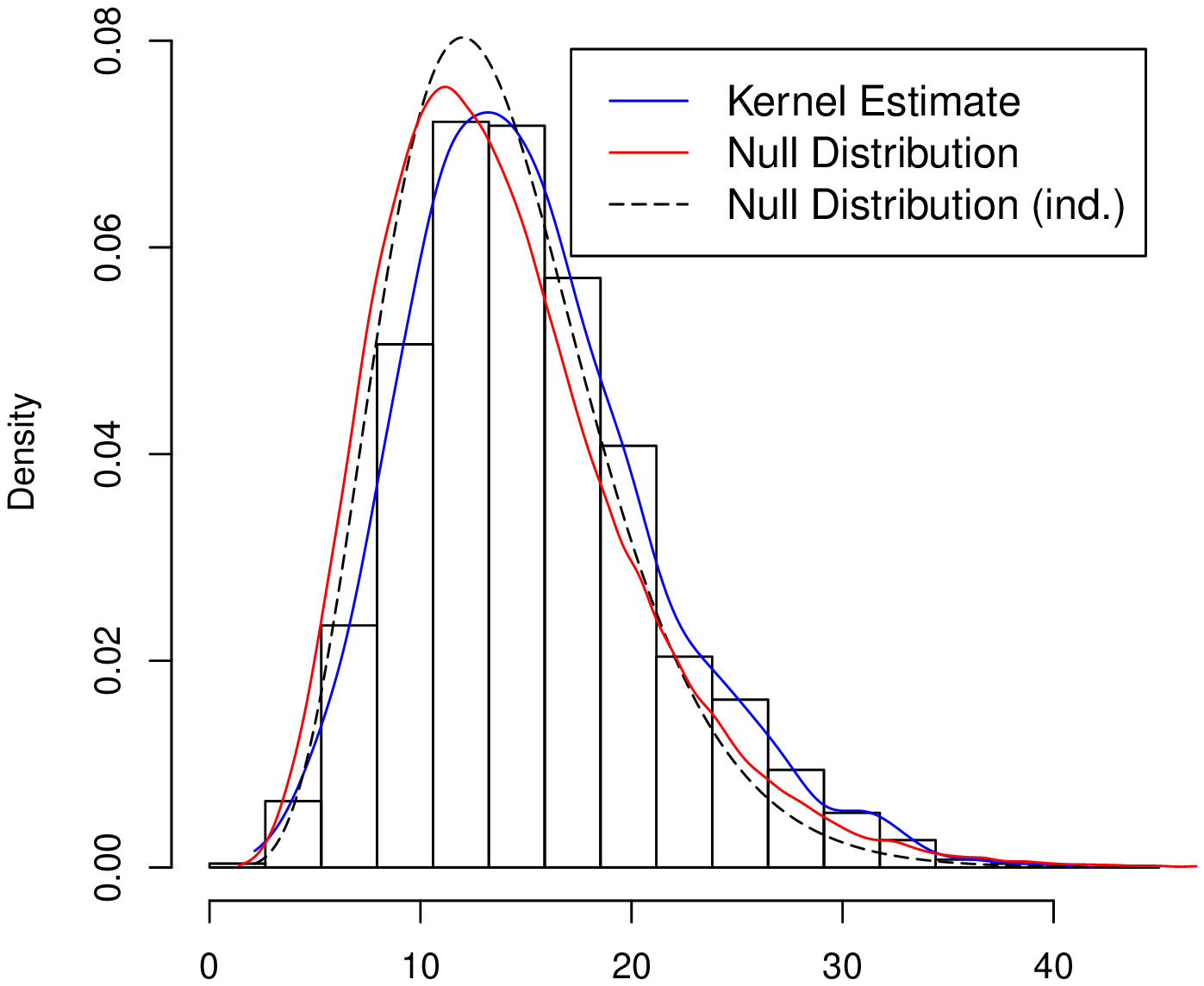} 
 \end{tabular}%
 \caption{The observed kernel density estimate of the GLR test, the fitted theoretical null distribution and its approximation assuming all predictors are independent. The error distribution is (a) N(0,1), and (b) $\chi^2(5)$.}
 \label{fig:null}
 \end{figure}

\subsection{Null Distribution}
Let us consider the nonparametric additive model defined in (\ref{model0}), where $P=5$ and $Q=4$. All five random variables in $\X$ and the first two random variables in $\Z$ are discrete. The number of discrete values taken by $X_1, X_2, \cdots, X_5$ and $Z_1, Z_2$ are 3, 4, 5, 4, 3 and 5, 4, respectively, starting from zero with an increment one. The probabilities for each of these variables are generated independently from a uniform (0,1) distribution, and then they are standardized so that the total probabilities become one. We have discarded very low values of probability ($<0.05$) to avoid very small or zero frequencies. To make the situation general, we have taken few independent and few dependent variables. $X_1$ and $X_2$ are independent random variables; $X_3, X_4, X_5$ form a group of dependent variables, but they are independent of $X_1$ and $X_2$.  Similarly, $(Z_1, Z_2)$ and $(Z_3, Z_4)$ are two independent groups.  The covariance matrices for the dependent groups are also generated randomly. Finally, these parameters are kept fixed throughout the entire simulation. To generate a set of dependent discrete variables, first, a random sample is drawn from a multivariate normal distribution with a fixed covariance matrix, then   observations are discretized based on their probabilities. Here, we are interested in testing the null hypothesis $H_0^{**}: m_1=m_2=\cdots = m_5=0$. To show the null distribution we have taken $m_p=0$ for all $p=1,2, \cdots, 5$. For  covariates, the $m_d$ functions are taken as
\be
m_6(Z_1) =  Z_1, \ \ m_7(Z_2) = Z_2^2, \ \ m_8(Z_3) = Z_3^2, \mbox{ and } m_9(Z_4) = sin(\pi Z_4).
\label{z_fun}
\ee
Notice that these functions are not centered at mean zero. However, it does not violate  assumptions of model (\ref{model0}) as  constants needed to center those functions  contribute to the intercept term $\alpha$. The smoother using Nadaraya-Watson estimator \citep{MR0185765} is taken to smooth $Z_3$ and $Z_4$. We used the default bandwidth parameter for the kernel density as $h = 1.06sn^{-1/5}$, where $s$ is the standard deviation of corresponding variable, and $n$ is the sample size. We have taken two different types of error distributions -- the standard normal distribution and the chi-square distribution with 5 degrees of freedom.  A sample of size 500 from $(Y, \X, \Z)$ is generated, and this exercise is replicated  1,000 times. The histograms of the observed GLR nonparametric test statistic  are presented in Figure \ref{fig:null}, and the corresponding kernel density estimates are also plotted. 
% 
% Theorem \ref{theorem:chi} states that the null distribution of the GLR test statistic depends on the eigenvalues of $\sum_{p=1}^P \bo{S}_p^*$. Therefore, we have generated 1,000 random sample from $\X$, and calculated the eigenvalues and corresponding theoretical densities for null distribution. The null distribution plotted in Figure \ref{fig:null} is the average of all these densities. 
%
The plots show that the empirical distributions match with the theoretical null distribution obtained from Corollary \ref{theorem:chi}. The plots give an indication of inflated level, but further simulation studies show that the convergence improves as sample size increases. In the same figure, we have  plotted the density of the null distribution of the GLR  test  under independence assumption on the predictors. It is a single chi-square distribution with degrees of freedom $\sum_{p=1}^P (k_p -1) = 14$. It  gives a good approximation of the null distribution. In fact, further simulation studies show that, unless some predictors are strongly correlated, this approximation works reasonably well. 
% As the actual null distribution depends on the value of the predictors, we have generated data keeping $\X$ fixed, and observed (not presented here) that the kernel density estimate fits well with its null distribution and the corresponding approximation. 
 In Figures \ref{fig:null}(a) and (b) the error distributions are different in two plots, so they  demonstrate that the null distribution of the GLR  test  does not depend on the choice of the error distribution. 

% \begin{figure}
% \centering%
% \begin{tabular}{cc}
% \includegraphics[height=7cm, width=7cm]{power_additive_n_500.eps} &
%  \includegraphics[height=7cm, width=7cm]{power_additive_n_500_yz_lin.eps} \\
% (a) & (b)\\
% \includegraphics[height=7cm, width=7cm]{power_additive_n_500_yx_lin.eps} &
%  \includegraphics[height=7cm, width=7cm]{power_additive_n_500_yx_lin_yz_lin.eps} \\
% (c) & (d)
% \end{tabular}%
% \caption{Simulated powers of the F-test and different GLR tests in the following situations -- (a) both $Y,X$ and $Y,Z$ are nonlinear, (b) $Y,X$ is nonlinear but $Y,Z$ linear, (c)  $Y,X$ is linear but $Y,Z$ nonlinear, and (d) both $Y,X$ and $Y,Z$ are linear. Black lines give the theoretical power functions calculated from Corollary \ref{corollary:power}.}
% \label{fig:power}
% \end{figure}

\begin{figure}
%\vspace{-.5cm}
\centering%
\begin{tabular}{cc}
\vspace{-.5cm} (a) & (b)\\
 \includegraphics[height=6cm, width=7.5cm]{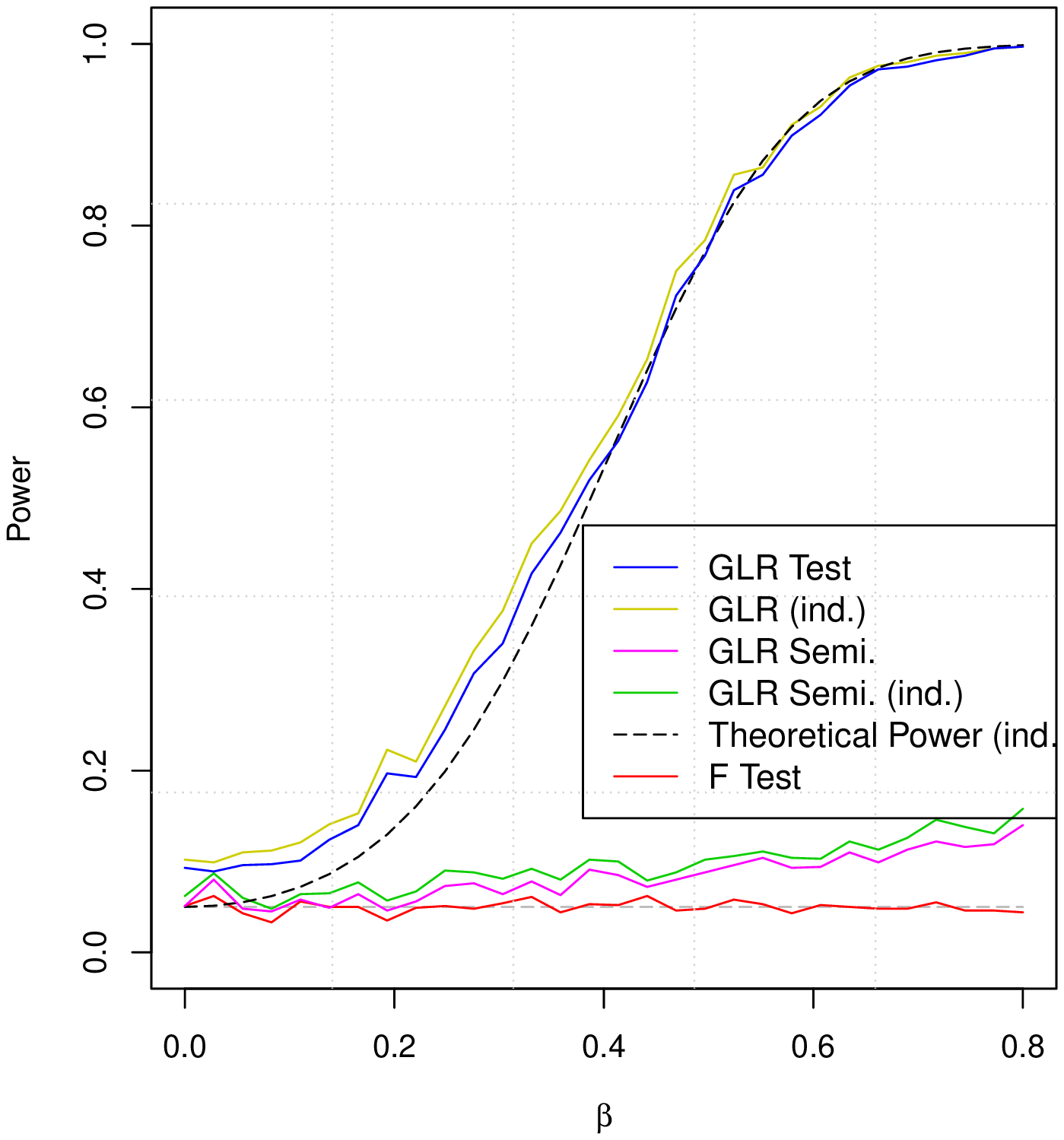}  &  \includegraphics[height=6cm, width=7.5cm]{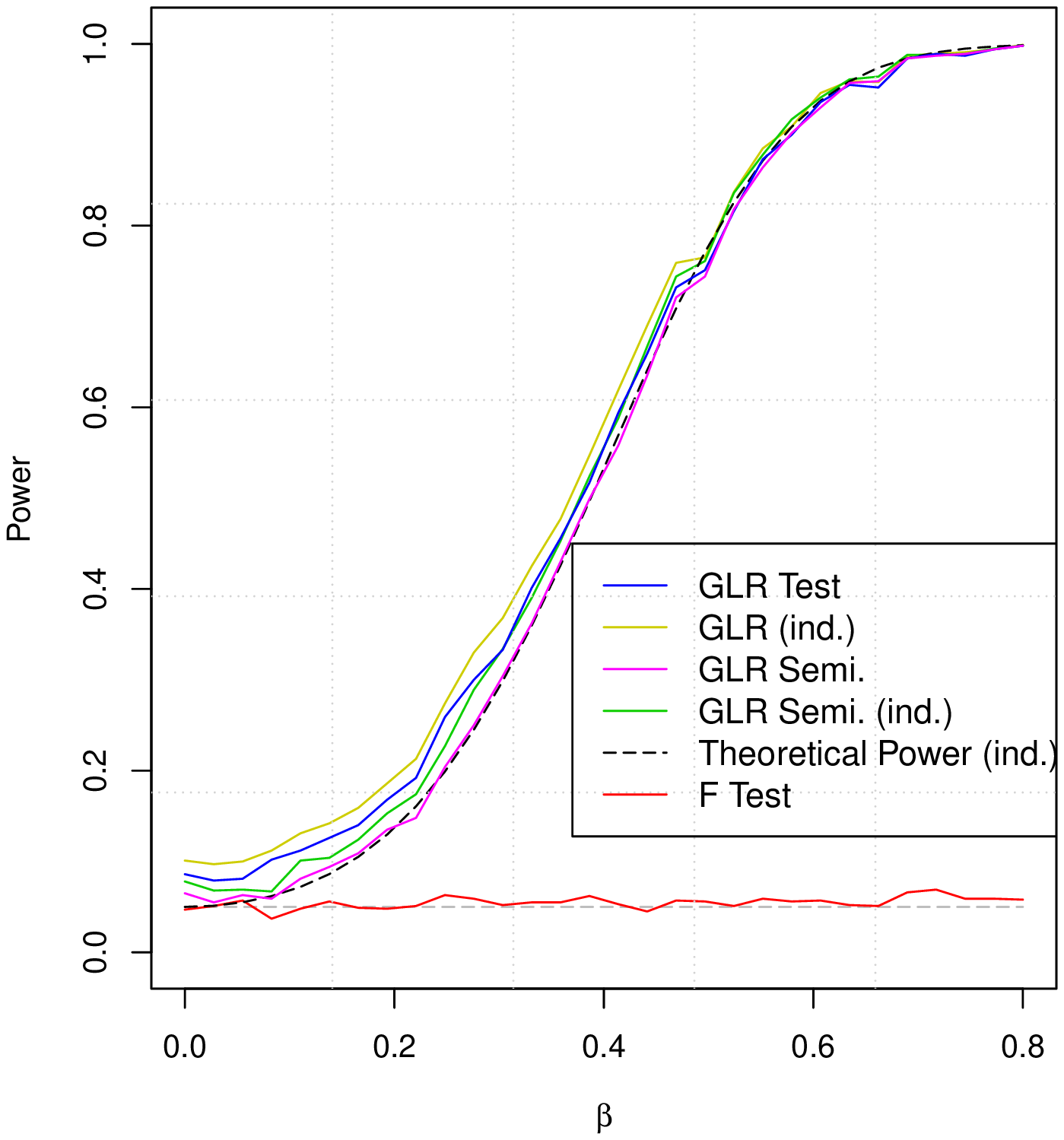} \\
\vspace{-.5cm} (c) & (d)\\
 \includegraphics[height=6cm, width=7.5cm]{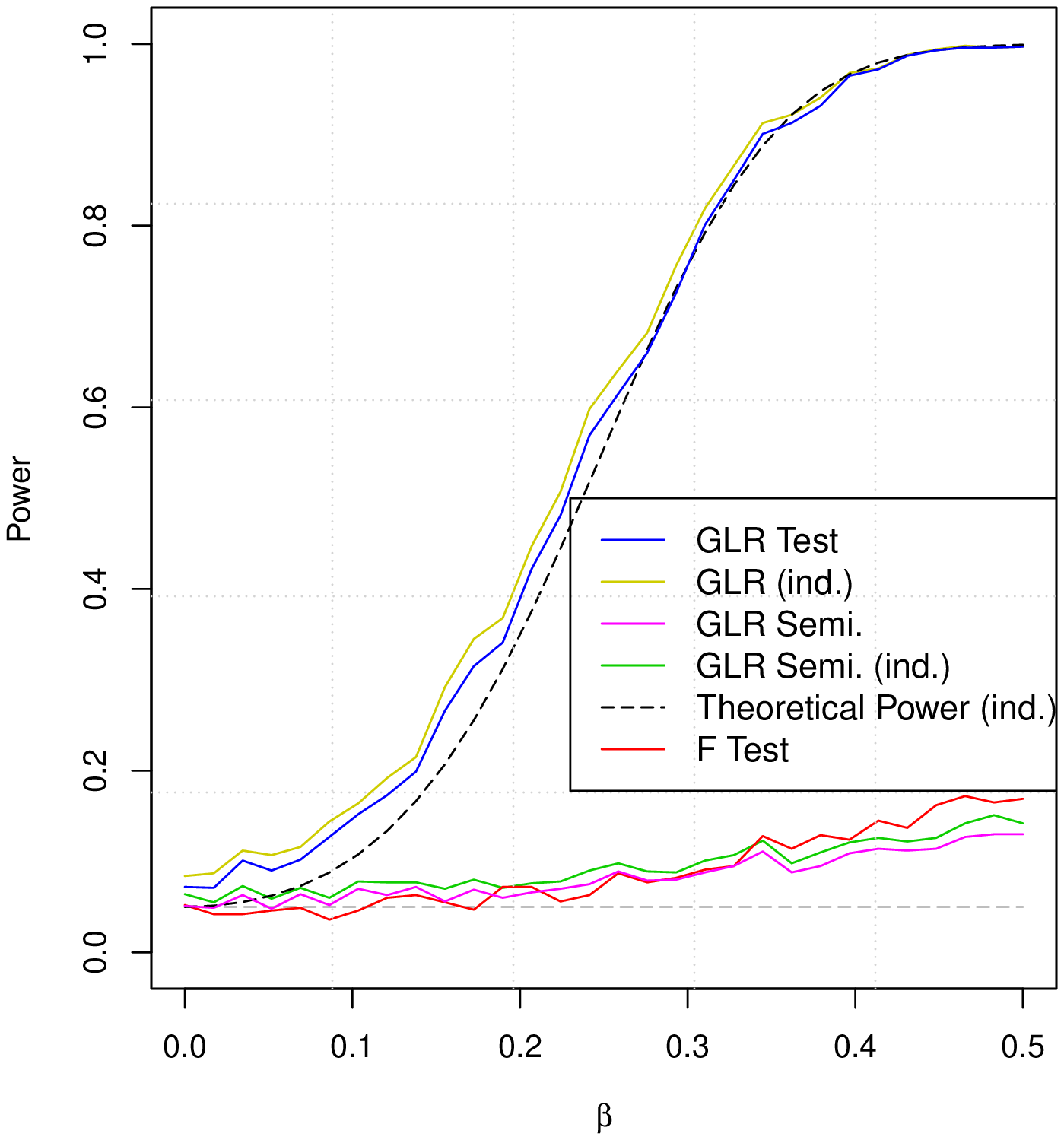} &
 \includegraphics[height=6cm, width=7.5cm]{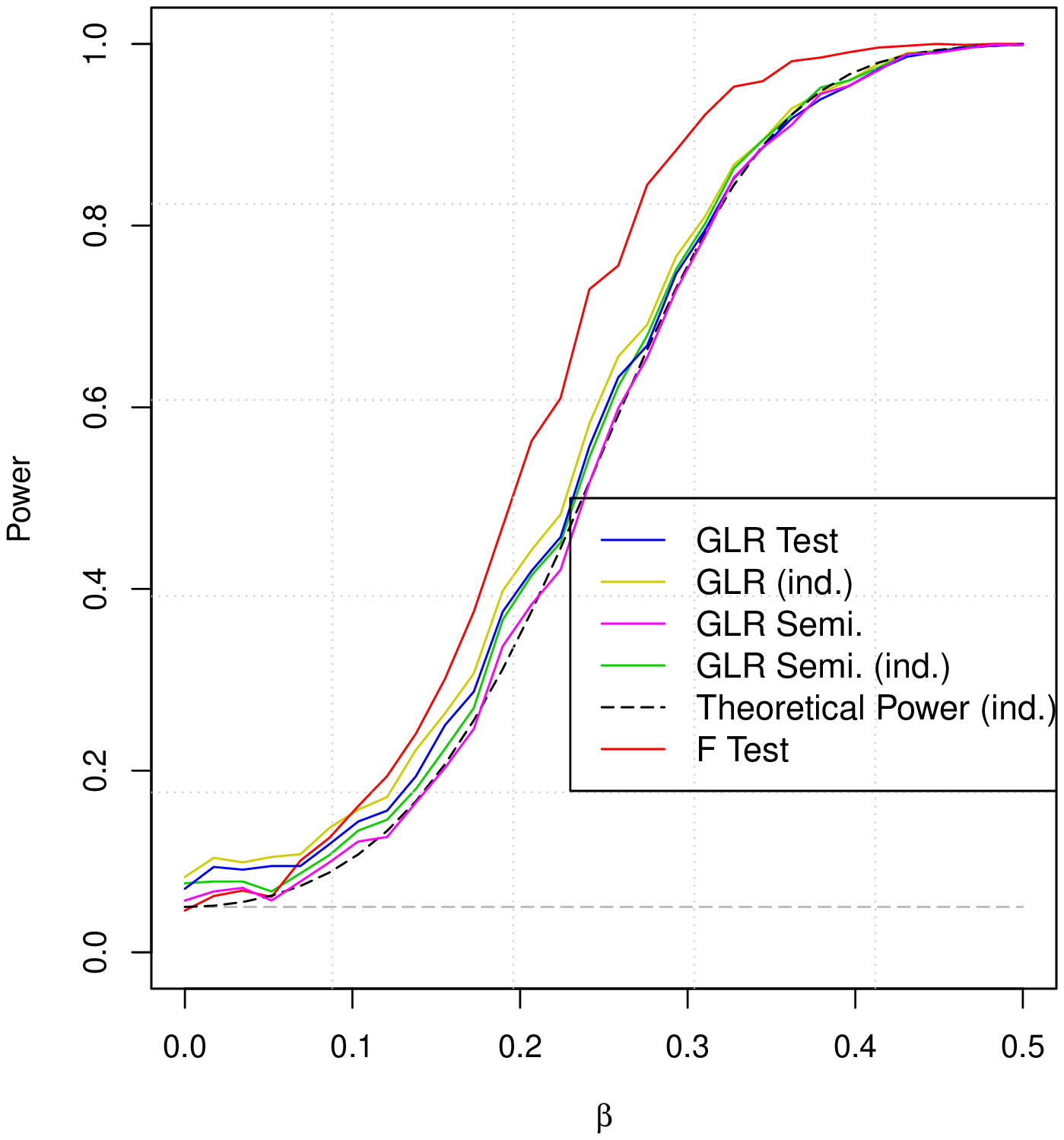} 
\end{tabular}%
\caption{Simulated powers of the F-test and different GLR tests when -- (a) both $Y,X$ and $Y,Z$ are nonlinear, (b) $Y,X$ is nonlinear but $Y,Z$ linear, (c)  $Y,X$ is linear but $Y,Z$ nonlinear, and (d) both $Y,X$ and $Y,Z$ are linear.}
\label{fig:power}
\end{figure}

\subsection{Power Function} 
In the power calculation for testing $H_0^{**}: m_1=m_2=\cdots = m_5=0$, we have taken the same setup of the previous example, except for the distribution of $X_1$ and the corresponding $m_1$ function. Here, we have taken
\be
 m_1(x_1) = \beta (x_1 - 0.75)^2,
 \label{m3}
\ee
and the distribution of $X_1$ is given by
% \bee
$P(X_1=0) = p^2, \ P(X_1=1) = 2pq \mbox{ and } P(X_1=2) = q^2,$
% \eee
where $p=1-q=0.75$. So, it violates the null hypothesis $H_0^{**}$ if $\beta \neq 0$. 
%The error term of the model is generated from the standard normal distribution. 
%However, we observed that the performance of the
% GLR test remains unchanged even if we take a different distribution. 
 As we compare the power of the GLR test and the classical F-test associated with the nested linear models, the error term is generated from the standard normal distribution to make all results comparable.  For each value of $\beta$, we  simulated 500 samples and  repeated it 1,000 times to calculate the observed power. The observed power   is the proportion of the test statistics greater than the corresponding critical value at 5\%  level of significance obtained from Corollary \ref{theorem:chi}. In Figure \ref{fig:power}(a), the GLR test shows a good power, but the F-test completely fails in this situation. The GLR test is slightly anti-conservative at null, which was also reflected in Figure \ref{fig:null}. In the same plot, we presented the observed powers of few other tests including the GLR test under independence assumption, abbreviated as GLR (ind.). Its performance is very similar to the original GLR test although predictors were not pairwise independent. We also plotted the theoretical power function of the GLR test calculated from Corollary \ref{corollary:power2} under the independence assumption. The observed and the theoretical power functions are very close to each other. 

If we fit a linear regression model in this setup, the breakdown situation of the F-test is apparent as the least square estimates of the regression coefficients vanish in the large sample sizes. Conditioning on the other variables the expected value of the least square estimate of the regression coefficient corresponding to $X_1$ turns out to be
\be
\beta_{LS} = \frac{Cov(X_1,Y)}{Var(X_1)} = - p m_1(0) +  (1 -2q) m_1(1)  +  q m_1(2),
\label{beta_ls}
\ee
which simplifies to zero for all values of $\beta$ in Equation (\ref{m3}) when $p=0.75$. Thus, $H_0^{**}$ seems to be true for all values of $\beta$ with respect to a linear model, and the power function for the F-test centers around the nominal level of the test. On the other hand, the full additive model successfully captures the nonlinear relationship; and the power of the GLR test tends to one as $\beta$ increases.

In the same plot, we presented the observed power function of the GLR semiparametric test (abbreviated as GLR Semi.) assuming that all covariates are linearly related with $Y$ (although it is not true), where the predictors are modeled nonparametrically. As the covariates are highly nonlinear, including a $\sin$ function, the semiparametric test breaks down. Notice that the distribution of the GLR test statistic does not change if the covariates are modeled parametrically, so the theoretical critical value of the GLR semiparametric test is same as the GLR nonparametric test, and it is obtained from Corollary \ref{theorem:chi}. The GLR Semi. (ind.) test in Figure \ref{fig:power}(a) is the approximation of the GLR semiparametric test under independence assumption on the predictor variables. The performance of these two tests are similar to the F-test. It shows that a wrong model of the covariates may  cause severe power loss.

We have investigated few more cases by generating different relationships between $Y,X$ and $Y,Z$. Figure \ref{fig:power}(b) presents  power functions where $Y,X$ is nonlinear as given in Equation (\ref{m3}), but $Y,Z$ is linear. The functions corresponding to the linear relationships between $Y$ and $Z$ are taken simply as $m_d(z) = z$ for $d=6,7,8,9$ instead of nonlinear functions in Equation (\ref{z_fun}). Here, we get the similar results from the GLR and F-test. The GLR semiparametric test gives almost equal power as described by its theoretical power function derived from Corollary \ref{corollary:power2}. In fact, in this case, the theoretical power of the GLR semiparametric test is same as the  GLR nonparametric test. But the advantage of semiparametric modeling is that its finite  sample performance is better than the full nonparametric test, if the modeling of the parametric part is correct. For this reason, the nonparametric GLR test is showing slightly inflated level, whereas the GLR semiparametric test properly maintains the level of the test.

In Figure \ref{fig:power}(c), we have plotted the power functions when $Y,X$ is linear by taking $m_1(x) = \beta x$ instead of Equation (\ref{m3}),  but $Y,Z$ is nonlinear as given in Equation (\ref{z_fun}). Even if the F-test successfully models the relationship between $Y$ and $X$, its power is almost unchanged as it fails to model the relationship between $Y$ and $Z$. Similarly, the GLR semiparametric test is  also showing poor power.  

Finally, the power functions of these tests, when both the relationships between $Y,X$ and $Y,Z$ are linear, are plotted in Figure \ref{fig:power}(d). Here, all assumptions of the F-test are satisfied, so it is the most powerful  among all unbiased tests.  It is interesting to notice that the power of all GLR tests are very competitive with the F-test. Therefore, even if the both relationships are linear, we do not expect to lose a significance amount of power by conducting the GLR test. These simulation results show that it is better to use the GLR test unless we are confident of a suitable parametric model. If some assumptions of the parametric model are violated, the F-test may break down. On the other hand, the GLR test produces very high power almost in all situations.

\subsection{The Goodness-of-Fit Test}
In a similar setup, we have studied the general GLR test including the goodness-of-fit testing problem for the semiparametric model. We test the null hypothesis (\ref{null3}) that there exists a linear relationship between $Y$ and $X_2$, but there is no effect of other components, i.e., $m_1=m_3=m_4=m_5=0$. We did not assume any parametric model for the covariates. In this simulation, we have taken $m_2(x_2) = \frac{1}{2} x_2$ and the same $m_1$ as given in Equation (\ref{m3}). So, the null hypothesis is true when $\beta=0$, and power of the GLR  test should increase as $\beta$ deviates from zero. The other setup for the simulation is same as the previous simulations including the functions for the covariates as given in Equation (\ref{z_fun}). Figure \ref{fig:gof}(a) shows that the observed and theoretical null distributions are close to each other.  The power functions of different tests are given in Figure \ref{fig:gof}(b). All tests in this plot are semiparametric tests, however, to make similarity with the previous simulation, we denote `GLR Semi.' when all covariates are assumed to be linear. `GLR Test' refers to the main semiparametric test whose null distribution is derived from Theorem \ref{theorem:chi3} without assuming that covariates are linearly modeled. Similarly `GLR (ind.)' is the approximation of this test by assuming that all predictors are pairwise independent. The black dotted line in the plot is the approximate power function derived for Corollary \ref{corollary:power2} that assumes  all predictors are pairwise independent. The main GLR test maintains the nominal level of the test and gives good power when the null hypothesis is not true. GLR (ind.) shows slight inflated power, however, it gives good approximation of the original test. GLR Semi. and GLR Semi. (ind.) fail to produce any significant power as the linearity assumptions on covariates are not satisfied.     
% We have conducted one more simulation study when covariates are linear. Here, we notice that all semiparametric tests are very powerful as expected. 

 \begin{figure}[t]
 \centering%
 \begin{tabular}{cc}
 \vspace{-.5cm} (a) & (b)\\
 \includegraphics[height=7cm, width=7.5cm]{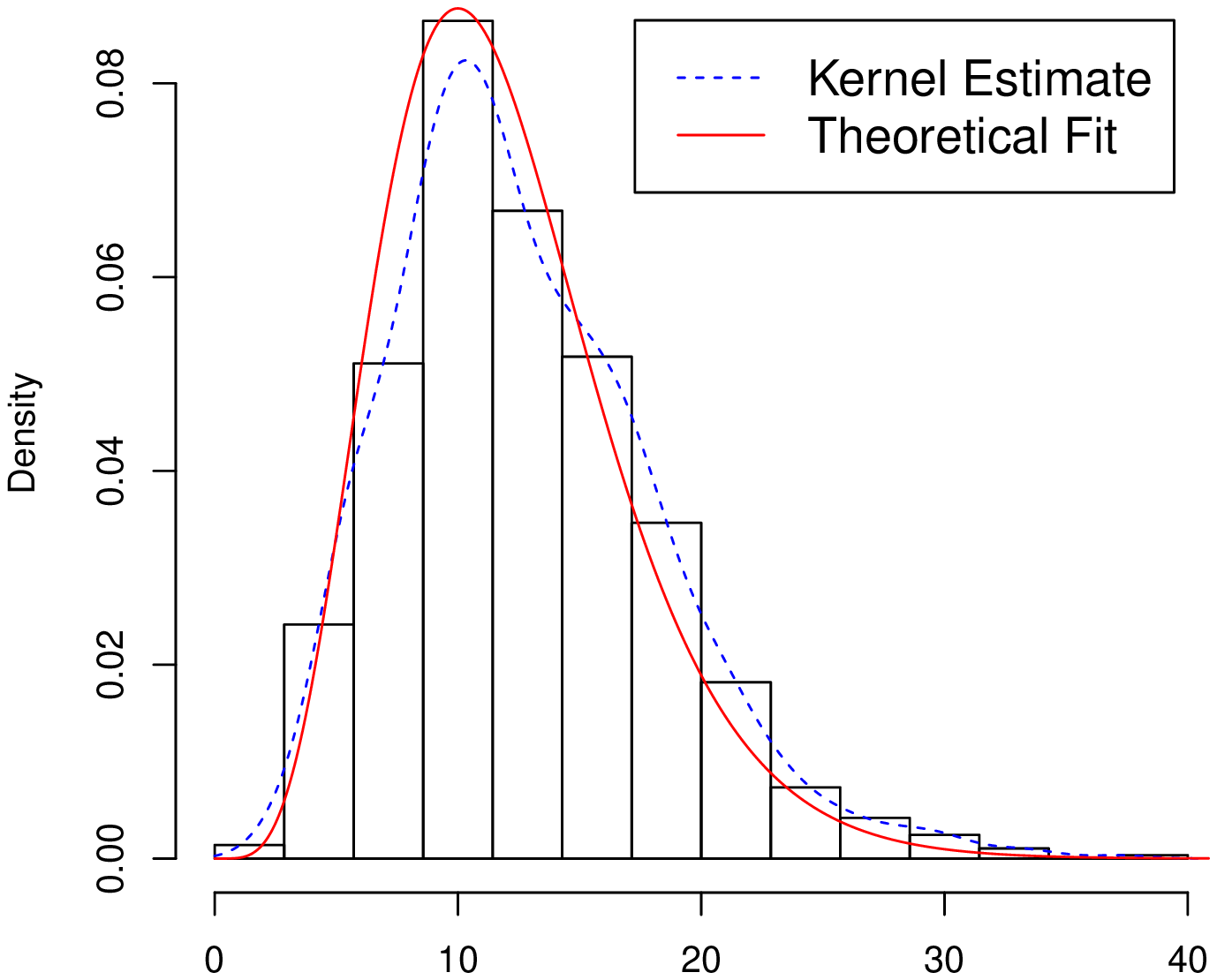} &
  \includegraphics[height=7cm, width=7.5cm]{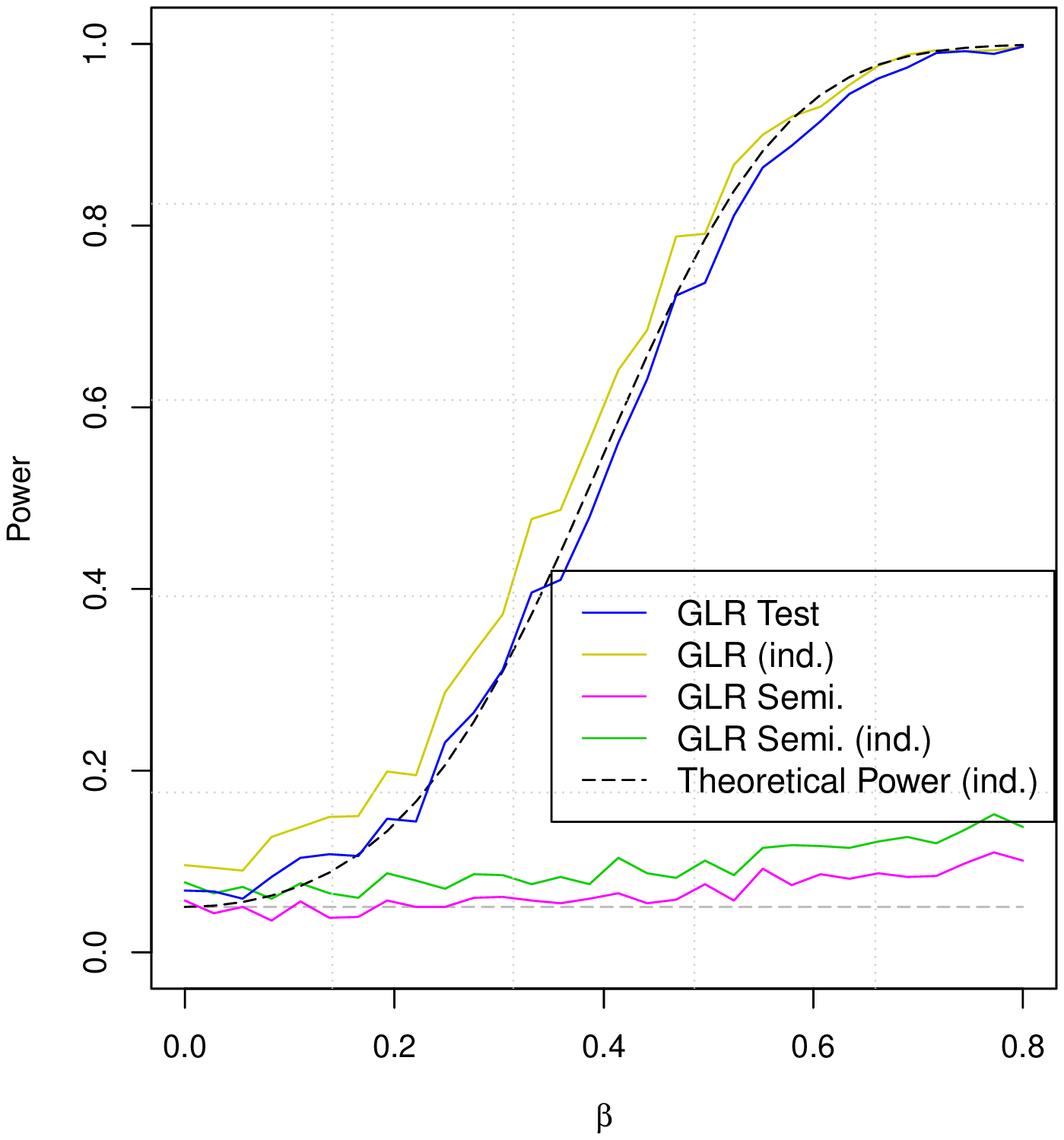} 
 \end{tabular}%
 \caption{(a) The observed and theoretical null distributions of the GLR semiparametric test statistic for the goodness-of-fit problem. (b) The power functions of different GLR semiparametric test statistics.}
 \label{fig:gof}
 \end{figure}

As a whole, these simulation results in Section \ref{sec:simulation} give enough justification for the theoretical results derived in the paper. These GLR tests are simple to calculate and produce a good power. As a virtue of the nonparametric method the tests do not depend on the error distribution of the model. For the model utility test, the GLR  test gives better power than the classical F-test when the parametric modeling is not appropriate. And even if the parametric model holds good, the GLR test produces a comparative power. The GLR test can be further  simplified if we assume that the predictors are pairwise independent. In this case, the test is slightly anti-conservative, but overall the approximation is good unless some predictors are strongly correlated. %The GLR tests are improved in terms of its finite sample performance if covariates are suitably modeled.  

%  \begin{figure}
% \centering%
% \begin{tabular}{cc}
% \vspace{-.5cm} (a) & (b) \\
% \includegraphics[height=5cm, width=7cm]{cut_raw.eps} &
% \includegraphics[height=5cm, width=7cm]{color_raw.eps} \\
% \vspace{-.5cm} (c) & (d)\\
% \includegraphics[height=5cm, width=7cm]{clarity_raw.eps} &
% \includegraphics[height=5cm, width=7cm]{carat_raw.eps} 
% \end{tabular}%
% \caption{Box plot or scatter plot of diamond price against (a) cut, (b) color, (c) clarity and (d) carat. In the box plots the middlemost lines of medians are replaced with the corresponding means.}
% \label{fig:raw}
% \end{figure}

\section{Real Data Example} \label{sec:real}
In this section, we apply the GLR test to analyze the diamonds data-set used in \cite{ggplot2}. This data-set contains the price (in 2008 US dollars) and other attributes of 53,940 diamonds. The attributes include the four C's of diamond quality -- cut, color, clarity and carat. There are three main  physical measurements $x$, $y$ and $z$  -- the largest length, width, and height of a diamond, respectively.  The data-set has other two physical measurements - depth and table, but we have not included them in this analysis as they are functions of $x$, $y$ and $z$. Carat is a unit of mass used for measuring gemstones and pearls. Cut is an objective measure of a diamond's light performance what we generally think of as sparkle. Cut, color and clarity are categorical variables, and other variables are continuous. There are five categories of cut - Fair, Good, Very Good, Premium and Ideal; and the percentage of each diamonds in this data-set are 2.98,  9.10, 22.30, 25.57 and 39.95, respectively. Color has seven categories D (best) to J (worst) with 12.56\%, 18.16\%, 17.69\%, 20.93\%, 15.40\%, 10.05\% and  5.21\%, respectively. Clarity contains  eight categories I1 (worst),  SI2,  SI1,  VS2,  VS1, VVS2, VVS1 and IF (best) with 1.37\%, 17.04\%, 24.22\%, 22.73\%, 15.15\%,  9.39\%,  6.77\% and 3.32\%, respectively.

\begin{figure}
	% \vspace{-.5cm}
	\centering%
	\begin{tabular}{cc}
		\vspace{-.5cm} (a) & (b) \\
		\includegraphics[height=5cm, width=7.5cm]{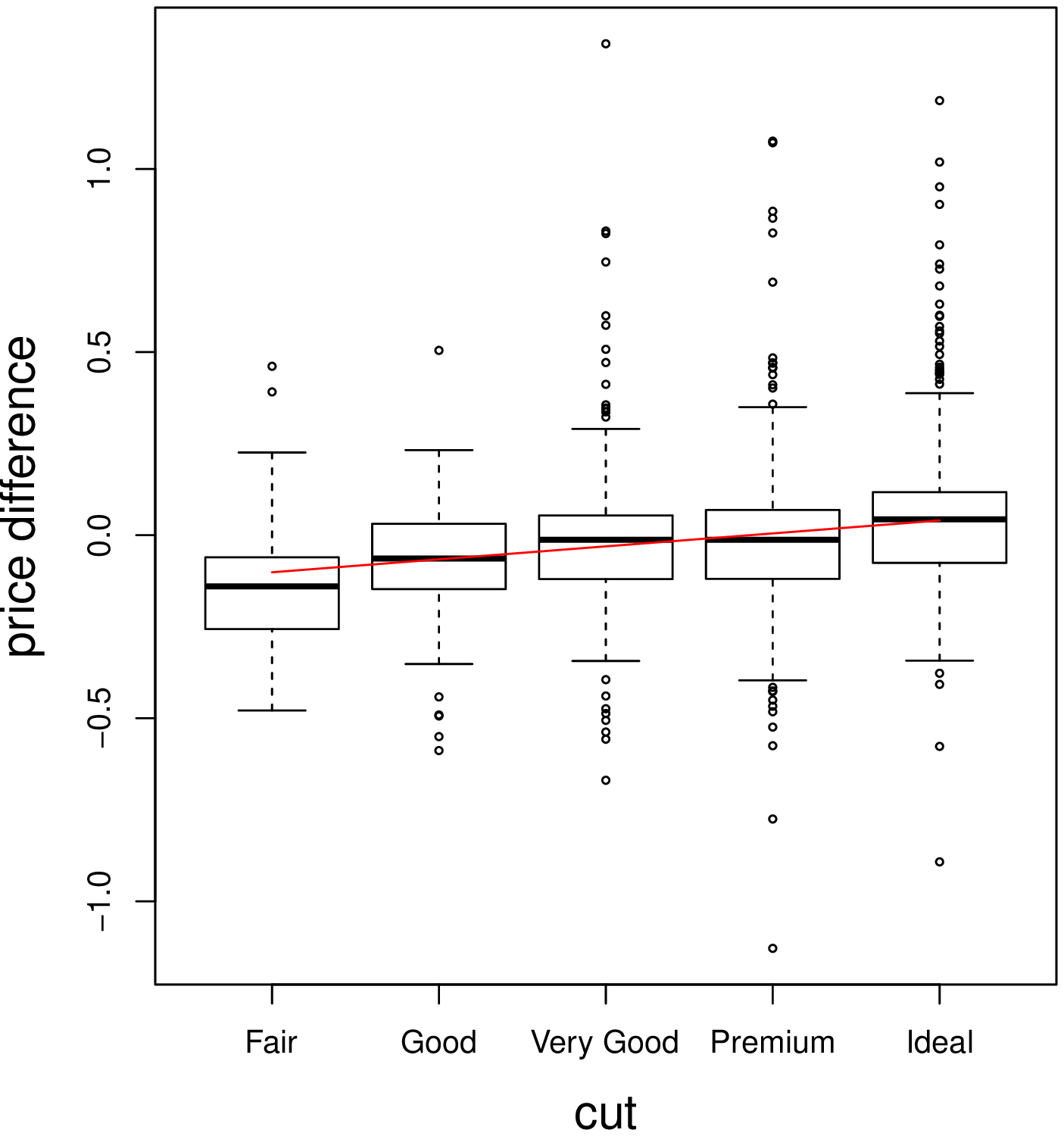} &
		\includegraphics[height=5cm, width=7.5cm]{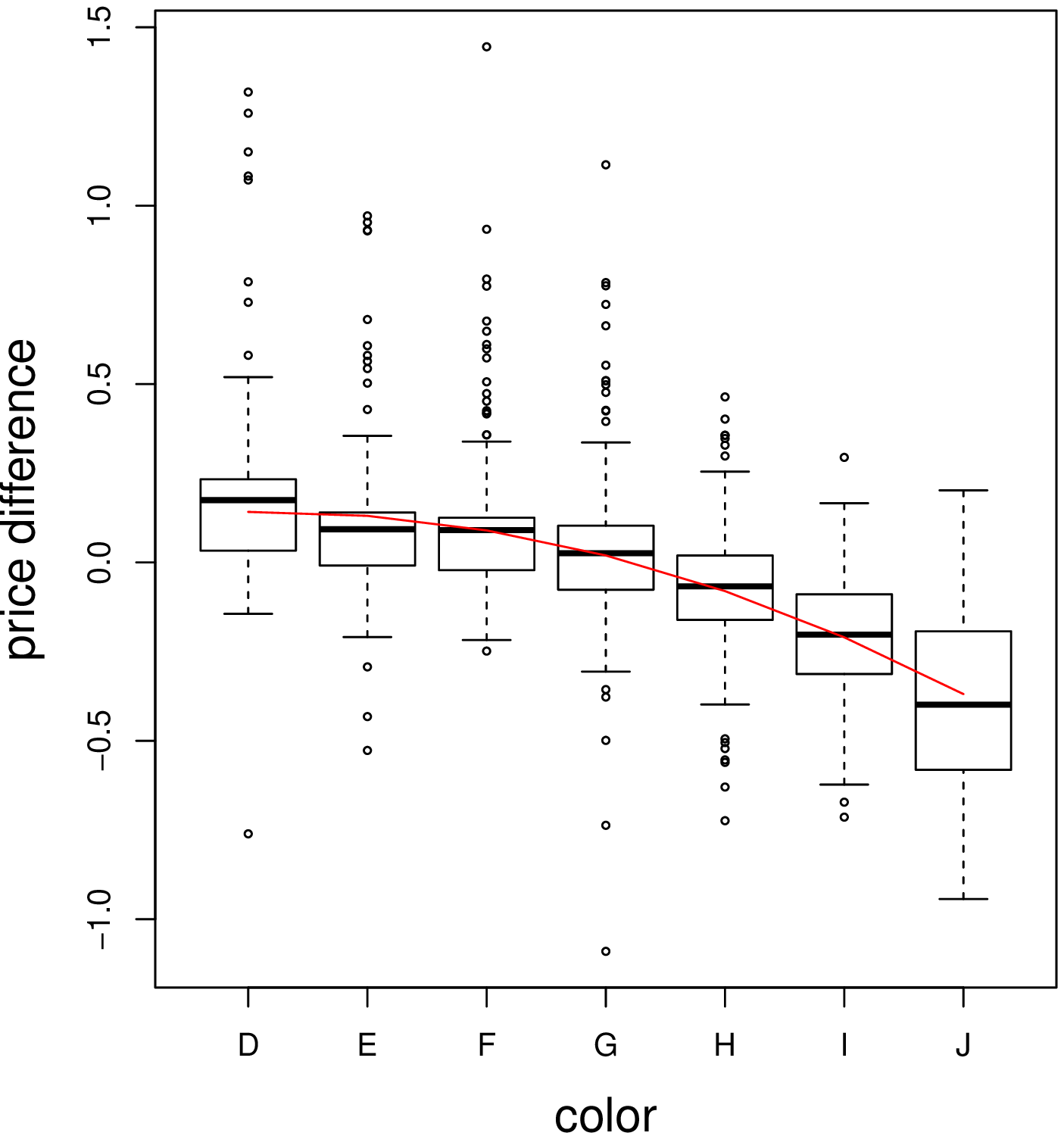} \\
		\vspace{-.5cm} (c) & (d) \\
		\includegraphics[height=5cm, width=7.5cm]{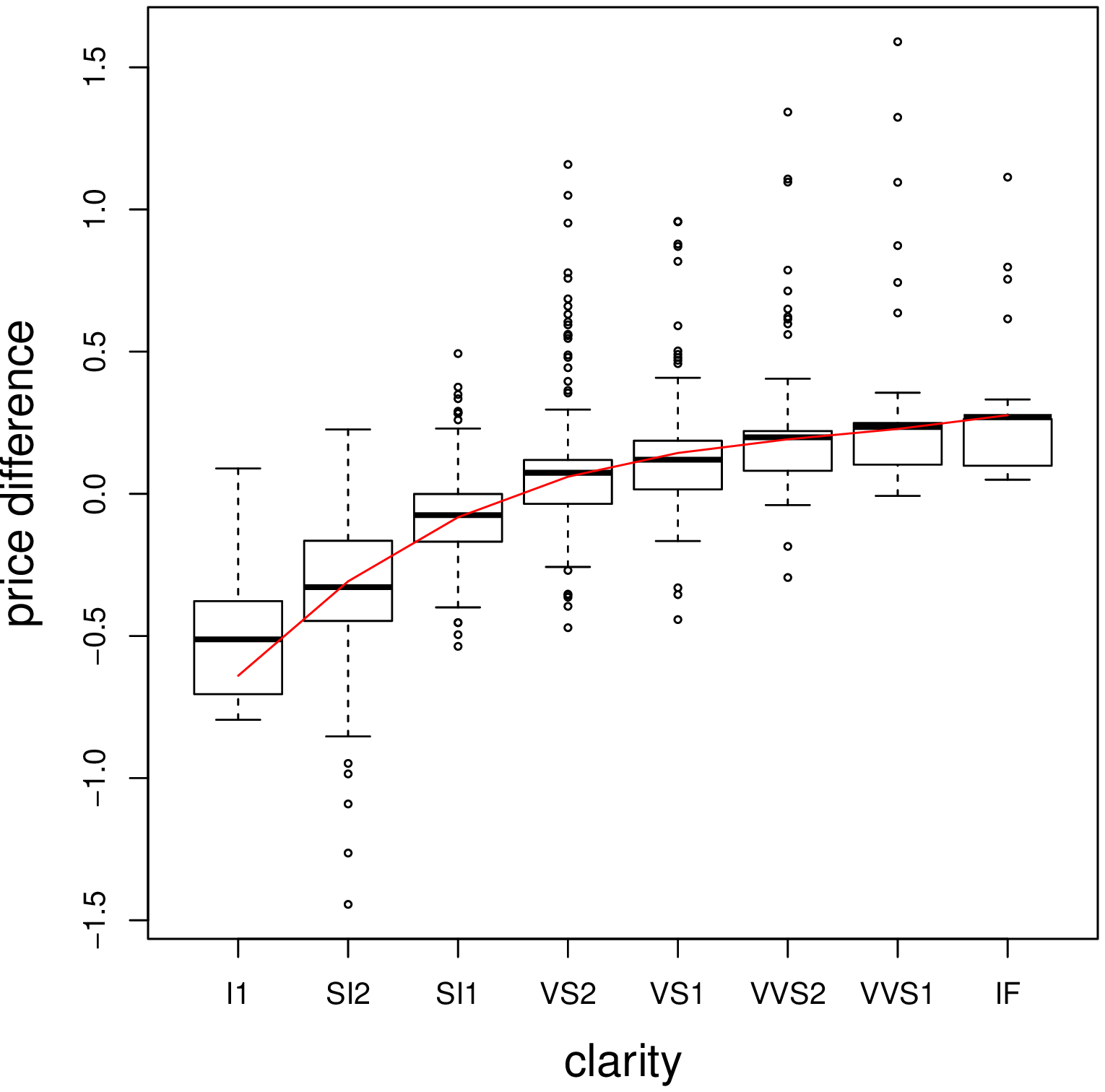} &
		\includegraphics[height=5cm, width=7.5cm]{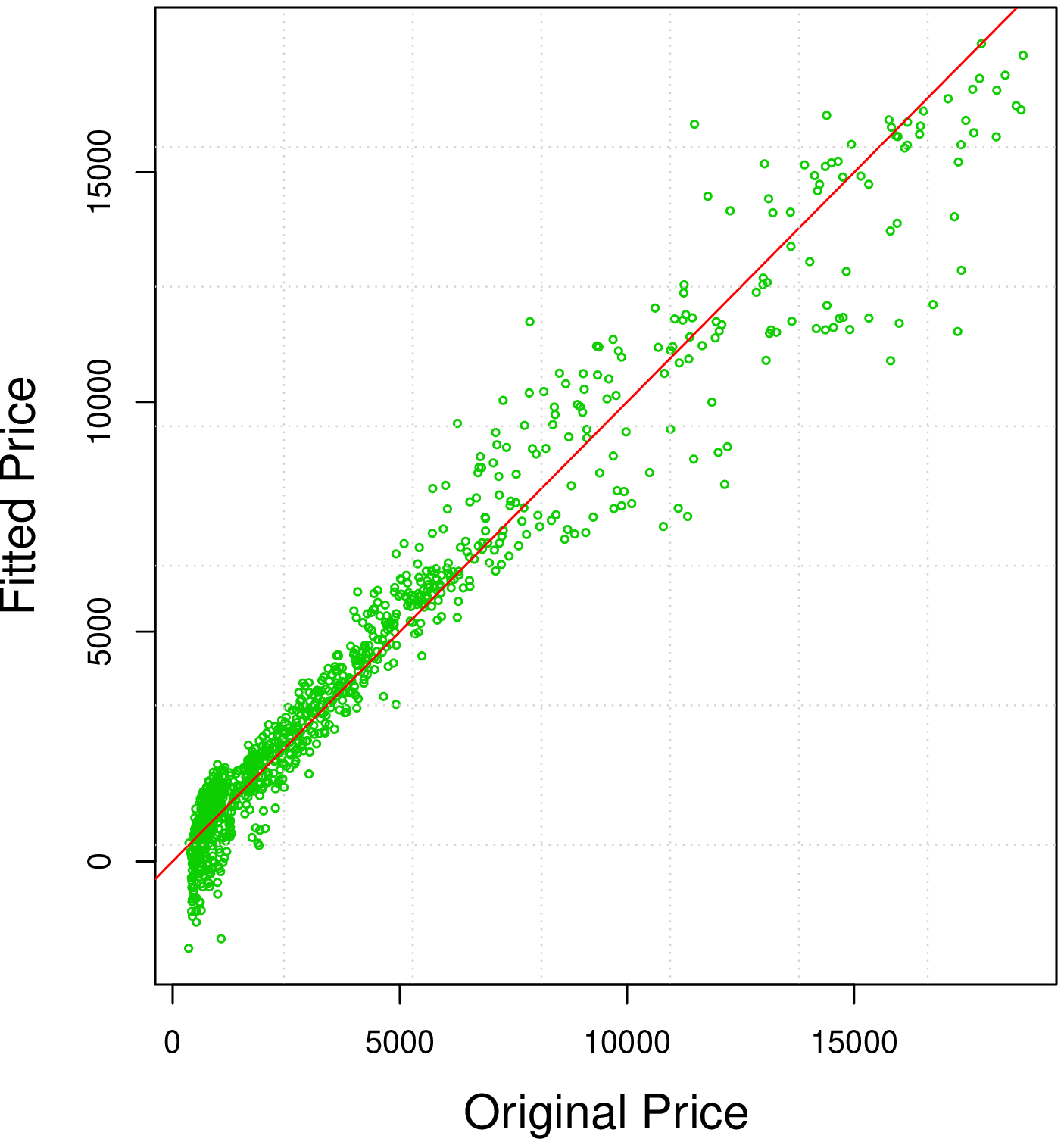} 
	\end{tabular}%
	\caption{Box plots showing the partial effect of (a) cut, (b) color and (c) clarity on the diamond price. The red line gives the fitted least square estimates, where $x$-axis is converted to the corresponding rank. (d) Plot of fitted price from the semiparametric model against the original price. The red line indicates the perfect fit.}
	\label{fig:dGLR}
\end{figure}

As cut, color and clarity are categorical variables, we  apply the GLR test for testing their effect in diamond price, whereas carat, $x$, $y$ and $z$  serve as covariates. The sample size is huge, so a slight deviation from a null hypothesis may cause rejection of the null hypothesis. For this reason, we have taken a random sub-sample of 1,000 observations from the data-set %Figure \ref{fig:raw} presents the bivariate plots of  price with cut, color, clarity and carat. In the box plots the middlemost lines of medians are replaced with corresponding means, whereas other lines represent their default meaning. The plots show that there are significant effect of these variables. However, box plots fail to show the ordinal effect among the categories of cut, color and clarity. In Figure \ref{fig:raw}(b) it seems that price of diamonds increases with category D to J, but in reality D is the best color and J is the worst. Thus the bivariate plots show that we must analyze the data by jointly considering all relevant attributes.
At first, we test the null hypothesis (\ref{null}) that there is no  effect of cut, color and clarity. So, $H_0^{**}: m_1 = m_2 =m_3 = 0$. The $p$-value of the GLR test comes out to be zero. Then, we have conducted three different tests to check the main effect of cut, color and clarity, separately, by taking others as covariates (for example,  $H_0^{**}: m_1 = 0$). Even in these cases $p$-values of the GLR tests are $6.45 \times 10^{-5}$ for cut and zero for color and clarity.  This is not surprising as there is a strong effect of cut, color and clarity in the price of diamonds.

All categorical variables, cut, color and clarity, maintain an order in their quality. So, in the next step, we have conducted some goodness-of-fit tests to know the effect of different levels of those variables.  We are particularly interested to know whether there is any specific pattern in the order of different levels of a category. The levels are ranked according to the increasing order of  quality. In our first goodness-of-fit test, we have taken the order of cut as a single predictor ($P=1$) and six other variables including the ranks of color and clarity as covariates ($Q=6$). We tested the null hypothesis that the diamond price is linear with the rank of the levels of cut. The $p$-value of the GLR test is 0.5476. This indicates that the diamond price may be modeled as a linear function of the rank of its cutting quality. The box plots in Figure \ref{fig:dGLR}(a) show the partial effect of cut after eliminating the effects of other variables. The middlemost line of median in the box plot is replaced by the corresponding mean, so it gives the mean effect of cut for that category. The red line in the plot is the least square regression line of the price difference on the rank of the cut effect. The plot also shows that the ranks in cut maintain a proper linear relationship in modeling the diamond price. %It also shows that the diamond price increases with the quality of cut, whereas the bivariate plot in Figure \ref{fig:raw}(a) failed to reflect it.

In the second semiparametric test, we checked the linearity effect of the rank of color in diamond price.
The GLR test gave a $p$-value of $2.02 \times 10^{-8}$, so we considered a quadratic model for the rank of color. The $p$-value of the test came out to be 0.0857, so the quadratic effect of color is not rejected at 5\% level of significance (see Figure \ref{fig:dGLR}(b)). Similarly, the $p$-values of corresponding to the linear, quadratic and cubic models of the rank of clarity are $10^{-15}$, 0.0412 and 0.4233, respectively. So, the effect of clarity in diamond price may be modeled as a cubic function of its rank (see Figure \ref{fig:dGLR}(c)). 

Finally, we have constructed a semiparametric model where the ranks of cut, color and clarity were taken as linear, quadratic and cubic functions, respectively. Carat, $x$, $y$ and $z$ serve as covariates for this test, and they are modeled nonparametrically. So, according to our notations $P=3$ and $Q=4$. The GLR  test (\ref{null2}) for the goodness-of-fit of the model gave a $p$-value of 0.1104 for this semiparametric model. Figure \ref{fig:dGLR}(d) displays the scatter plot of the original price of  diamonds and their fitted price from the semiparametric model; the correlation between the two prices is 0.9632. The points in the plot cluster around the red line which represents the perfect fit. These results indicate that if we build a semiparametric model for this data the appropriate choices of the additive components for the rank of cut, color and clarity are  a linear, quadratic and cubic functions, respectively.

The estimate of the overall mean price for  the data-set is $\hat\alpha=3931.3760$, and the standard deviation is $\hat\sigma = 4149.5376$. Let us consider an additive model as
$y = \hat\alpha + \hat\sigma \sum_{d=1}^7 m_d, $ where $m_1, m_2, \cdots, m_7$ being the additive effects of the rank of cut, color, clarity, carat, $x$, $y$ and $z$, respectively. Then, the parametric parts of the fitted semiparametric model are given by
\ba
m_1(x_1) &=&  -0.1330 +   0.0344 x_1, \ \ \ \ 
m_2(x_2) =   0.1271 +   0.0312 x_2   - 0.0146 x_2^2, \nn
m_3(x_3) &=& -1.1149 +   0.5466 x_3  - 0.0801 x_3^2  +    0.0042 x_3^3  .
\label{mp_est}
\ea

 \section{Discussion} \label{sec:concluding}
In an  additive model, a novel method is derived for testing the main effect of the predictor variables which take discrete or categorical values. The main effect is adjusted by  covariates possibly containing continuous valued random variables. The predictors and covariates are modeled nonparametrically using an additive model, so the test avoids loss of power due to model misspecification often arises in classical parametric tests. This method is further extended to the semiparametric model, and a goodness-of-fit test is derived. The simulation results show that the GLR test may outperform the parametric test when the model for just one of the components fails; and at the same time, it produces a comparable power as the conventional tests if the assumed parametric model holds good. The power of the GLR test is asymptotically optimal in terms of the rate of convergence, and it can detect a specific class of contiguous alternatives at a rate of $n^{-1/2}$. In case of categorical predictors the GLR test generalizes the classical ANCOVA by modeling covariates nonparametrically and without assuming normality of the error term. So, it is a development of the basic statistical theory, and the methodology can be widely useful in practice. %A potential use of the test is in GWAS as the predictor SNPs are all categorical  variables.
 
%  The GLR test statistic is constructed using bin smoothers and local polynomials, therefore, it is not robust in general. But the test may be robustified by replacing the residual sum of squares by their robust versions. We have assumed that errors of the model are homoscedastic, however, it  may be suitably modified for the heteroscedastic case. 

 \section*{Appendix A: Regularity Conditions}
 
% \noindent
%  {\bf{Regularity Conditions:}}
%  
%  \smallskip
%  \noindent
 To derive the asymptotic distribution of the GLR test statistic, we need the following assumptions: 
\begin{enumerate}
\item[(C1)] Suppose $c_{pj} = P(X_p=x_{pj})$, then $c_{pj} \in (0,1)$   for all $j=1,2, \cdots, k_p$ and $p=1,2, \cdots, P$, where   $\sum_{j=1}^{k_p} c_{pj} =1$.

 \item[(C2)] The kernel function $K(z)$ is bounded and Lipschitz-continuous
with a bounded support.

\item[(C3)] If $Z_q$ is continuous, then the density $f_q$ of $Z_q$ is Lipschitz-continuous and bounded away from 0 and has bounded supports $\Omega_q$ for $q \in \{1,2, \cdots , Q\}$.

\item[(C4)] If both $Z_q$ and $Z_{q'}$ are continuous, then the joint density $f_{qq'}$ of $Z_q$ and $Z_{q'}$ is Lipschitz-continuous on its support $\Omega_q \times \Omega_{q'}$ for $q \neq q' \in \{1,2, \cdots , Q\}$.

\item[(C5)] $nh_q / \log(n) \rightarrow \infty$ as $n \rightarrow \infty$ and $h_q \rightarrow 0$   for $q = 1,2, \cdots , Q$.

\item[(C6)] If $Z_q$ is continuous and $d_q$ be the degree of the polynomial used for smoothing of $Z_q$, then the $( d_q + 1)$-th derivative of $m_{P+q}$, for $q \in \{1,2, \cdots , Q\}$, exist and is bounded and continuous.

\item[(C7)] $\sigma^2 = \mbox{Var}[\epsilon] = E[\epsilon^2 ] < \infty$.

\item[(C8)]  $E[ m_q (Z_q)| X_p = x_{pj} ] = 0$ for all $j=1,2,\cdots, k_p$, $ p = 1,2, \cdots , P$  and $q = 1,2, \cdots , Q$.

%\item[(C9)] $X_1, X_2, \cdots, X_P$ are linearly independent.
\end{enumerate}

%\par
%%%%%%%%%%%%%%%%%%%%%%%%%%%%%%%%%%%%%%%%%%%%%%%%%%%%%%%%%%%%%%%%%%%%%%%%%%%%%%%%%%%%%%%%%%%%%%%%%%%%%%%%%%%%%%%%%%%%%%%%%%%%
\vskip 14pt
\noindent {\large\bf Acknowledgements:}
%\begin{acknowledgements}
	The author very much appreciates Kuchibhotla Arun Kumar for carefully reading the paper including all proofs and providing helpful comments and suggestions. %The author would like to thank two anonymous referees which significantly improved the presentation of the paper.
	
% The author would like to thank  Kuchibhotla Arun Kumar for his helpful comments and suggestions. The author gratefully acknowledge the comments of two anonymous referees which significantly improved the presentation of the paper.
%\par
%\end{acknowledgements}
%%%%%%%%%%%%%%%%%%%%%%%%%%%%%%%%%%%%%%%%%%%%%%%%%%%%%%%%%%%%%%%%%%%%%%%%%%%%%%%%%%%%%%%%%%%%%%%%%%%%%%%%%%%%%%%%%%%%%%%%%%%%

% \bibliographystyle{abbrvnat}
\bibliographystyle{spbasic}
 \bibliography{reference}

\begin{thebibliography}{26}
\providecommand{\natexlab}[1]{#1}
\providecommand{\url}[1]{{#1}}
\providecommand{\urlprefix}{URL }
\expandafter\ifx\csname urlstyle\endcsname\relax
  \providecommand{\doi}[1]{DOI~\discretionary{}{}{}#1}\else
  \providecommand{\doi}{DOI~\discretionary{}{}{}\begingroup
  \urlstyle{rm}\Url}\fi
\providecommand{\eprint}[2][]{\url{#2}}

\bibitem[{Bapat(2012)}]{bapat2012linear}
Bapat RB (2012) Linear algebra and linear models. Springer Science \& Business
  Media

\bibitem[{Buja et~al(1989)Buja, Hastie, and Tibshirani}]{MR994249}
Buja A, Hastie T, Tibshirani R (1989) Linear smoothers and additive models. Ann
  Statist 17(2):453--555

\bibitem[{Davies(1980)}]{davies1980algorithm}
Davies RB (1980) The distribution of a linear combination of $\chi^2$ random
  variables. {A}lgorithm {AS}155. Appl Statist 29:323--333

\bibitem[{Fan and Jiang(2005)}]{MR2201017}
Fan J, Jiang J (2005) Nonparametric inferences for additive models. J Amer
  Statist Assoc 100(471):890--907

\bibitem[{Fan et~al(2001)Fan, Zhang, and Zhang}]{MR1833962}
Fan J, Zhang C, Zhang J (2001) Generalized likelihood ratio statistics and
  {W}ilks phenomenon. Ann Statist 29(1):153--193

\bibitem[{Friedman and Stuetzle(1981)}]{MR650892}
Friedman JH, Stuetzle W (1981) Projection pursuit regression. J Amer Statist
  Assoc 76(376):817--823

\bibitem[{Hall and Marron(1988)}]{MR970470}
Hall P, Marron JS (1988) Variable window width kernel estimates of probability
  densities. Probab Theory Related Fields 80(1):37--49

\bibitem[{Hastie and Tibshirani(2000)}]{MR1820768}
Hastie T, Tibshirani R (2000) Bayesian backfitting (with discussion). Statist
  Sci 15(3):196--223, with comments and a rejoinder by the authors

\bibitem[{Hastie and Tibshirani(1990)}]{MR1082147}
Hastie TJ, Tibshirani RJ (1990) Generalized additive models, Monographs on
  Statistics and Applied Probability, vol~43. Chapman and Hall, Ltd., London

\bibitem[{Ingster(1993)}]{MR1257978}
Ingster YI (1993) Asymptotically minimax hypothesis testing for nonparametric
  alternatives. {I--III}. Math Methods Statist 2(2):2:85--114; 3:171--189;
  4:249--268

\bibitem[{Jiang et~al(2007)Jiang, Zhou, Jiang, and Peng}]{MR2396026}
Jiang J, Zhou H, Jiang X, Peng J (2007) Generalized likelihood ratio tests for
  the structure of semiparametric additive models. Canad J Statist
  35(3):381--398

\bibitem[{Mammen et~al(1999)Mammen, Linton, and Nielsen}]{MR1742496}
Mammen E, Linton O, Nielsen J (1999) The existence and asymptotic properties of
  a backfitting projection algorithm under weak conditions. Ann Statist
  27(5):1443--1490

\bibitem[{Opsomer(2000)}]{MR1763322}
Opsomer JD (2000) Asymptotic properties of backfitting estimators. J
  Multivariate Anal 73(2):166--179

\bibitem[{Opsomer and Ruppert(1997)}]{MR1429922}
Opsomer JD, Ruppert D (1997) Fitting a bivariate additive model by local
  polynomial regression. Ann Statist 25(1):186--211

\bibitem[{Opsomer and Ruppert(1998)}]{MR1631333}
Opsomer JD, Ruppert D (1998) A fully automated bandwidth selection method for
  fitting additive models. J Amer Statist Assoc 93(442):605--619

\bibitem[{Opsomer and Ruppert(1999)}]{opsomer1999root}
Opsomer JD, Ruppert D (1999) A root-n consistent backfitting estimator for
  semiparametric additive modeling. J Comput Graph Statist 8(4):715--732

\bibitem[{Speckman(1988)}]{MR970977}
Speckman P (1988) Kernel smoothing in partial linear models. J Roy Statist Soc
  Ser B 50(3):413--436

\bibitem[{Sperlich et~al(2002)Sperlich, Tj{\o}stheim, and Yang}]{MR1891823}
Sperlich S, Tj{\o}stheim D, Yang L (2002) Nonparametric estimation and testing
  of interaction in additive models. Econometric Theory 18(2):197--251

\bibitem[{Spokoiny(1996)}]{MR1425962}
Spokoiny VG (1996) Adaptive hypothesis testing using wavelets. Ann Statist
  24(6):2477--2498

\bibitem[{Stone(1985)}]{MR790566}
Stone CJ (1985) Additive regression and other nonparametric models. Ann Statist
  13(2):689--705

\bibitem[{Stone(1986)}]{MR840516}
Stone CJ (1986) The dimensionality reduction principle for generalized additive
  models. Ann Statist 14(2):590--606

\bibitem[{Tj{\o}stheim and Auestad(1994)}]{MR1310230}
Tj{\o}stheim D, Auestad BH (1994) Nonparametric identification of nonlinear
  time series: projections. J Amer Statist Assoc 89(428):1398--1409

\bibitem[{Wand(1999)}]{MR1701398}
Wand MP (1999) A central limit theorem for local polynomial backfitting
  estimators. J Multivariate Anal 70(1):57--65

\bibitem[{Watson(1964)}]{MR0185765}
Watson GS (1964) Smooth regression analysis. Sankhy\=a Ser A 26:359--372

\bibitem[{Wickham(2009)}]{ggplot2}
Wickham H (2009) ggplot2: elegant graphics for data analysis. Springer New York

\bibitem[{Yang et~al(2003)Yang, Sperlich, and H{\"a}rdle}]{MR1985882}
Yang L, Sperlich S, H{\"a}rdle W (2003) Derivative estimation and testing in
  generalized additive models. J Statist Plann Inference 115(2):521--542

\end{thebibliography}
 
 %%%%%%%%%%%%%%%%%%%%%%%%%%%%%%%%%%%%%%%%%%%%%%%%%%%%%%%%%%%%%
\clearpage

\appendix
\setcounter{section}{1}
\setcounter{equation}{0} %-1

%\vskip 14pt
%\noindent

\makeatletter\@addtoreset{equation}{section}
\def\theequation{\thesection.\arabic{equation}}
%\appendix\newpage\markboth{Appendix}{Appendix}
%\renewcommand{\thesection}{\Alph{section}}
%\numberwithin{equation}{section}

\begin{center}
	{\Huge{\bf Supplementary Materials}}
\end{center}

\section*{Appendix B: Backfitting Estimators}
%\section{Appendix B: Backfitting Estimators}
Let us define $\bo{m}_p = (m_p(X_{p1}), m_p(X_{p2}), \cdots, m_p(X_{pn}))^T$ for $p \in \{1,2, \cdots, P\}$, and $\bo{m}_{P+q} = (m_{P+q}(Z_{q1}), $ $ m_{P+q}(Z_{q2}) \cdots, m_{P+q}(Z_{qn}))^T$ for $q \in \{1,2, \cdots, Q\}$. The additive components, $\bo{m}_1, \bo{m}_2, \cdots, \bo{m}_{P+Q}$,  are estimated using the backfitting estimators. The first step is to select a suitable smoothing matrix $\bo{S}_d$ for $d \in \{1,2, \cdots, P+Q\}$, where $\widehat{\bo{m}}_d =\bo{S}_d\Y_{res}$ is the estimator of $\bo{m}_d$, and $\Y_{res}$ is the residual of $\bo{Y} = (Y_1, \cdots, Y_n)^T$ given other additive components. This step is then repeated until convergence of all additive components  \citep{MR1082147}. As $\X$ contains categorical or discrete valued random variables, the bin smoother at a point mass is appropriate. Suppose, for $j=1,2, \cdots, k_p$, there are $n_{pj}$  observations at $X_p= x_{pj}$, where $\sum_{j=1}^{k_p} n_{pj} =n$, $p=1, 2, \cdots, P$. If the observations are sorted according to the values of $X_p$, then the smoothing matrix for $m_p$ is given by
\be
\bo{S}_p = 
\begin{bmatrix}
    n_{p1}^{-1}\bo{J}_{n_{p1}} & \bo{O}_{n_{p1}, n_{p2}} & \dots  & \bo{O}_{n_{p1}, n_{pk_p}} \\
    \bo{O}_{n_{p2}, n_{p1}} & n_{p2}^{-1} \bo{J}_{n_{p2}} & \dots  & \bo{O}_{n_{p2}, n_{pk_p}} \\
    \vdots & \vdots  & \ddots & \vdots \\
    \bo{O}_{n_{pk_p},n_{p1}} & \bo{O}_{n_{pk_p},n_{p2}} & \dots  & n_{pk_p}^{-1} \bo{J}_{n_{pk_p}}
\end{bmatrix}, \mbox{ for } p = 1, 2, \cdots, P,
\label{sp}
\ee
where $\bo{J}_n$ is a $n\times n$ matrix with elements 1, and $\bo{O}_{m,n}$ is a $m\times n$ matrix with elements 0. It essentially means that  $\bo{S}_p$ is constructed such a way that $\widehat{\bo{m}}_p(x_{pj})$ is the partial mean of $\Y_{res}$, where  $X_p= x_{pj}$ for $p=1, 2, \cdots, P$ and $j=1,2, \cdots, k_p$.

The covariate $\Z$ may contain any type of random variable -- categorical,  discrete or continuous. If some components of $\Z$ are categorical or discrete, then we use bin smoother again. Otherwise, for continuous valued covariates, one may choose a smoother that uses local polynomials.  For the simplicity of notation, we assume that all covariates are continuous. In fact, the situation is even simpler for categorical or discrete covariates. Let $d_q$ be the degree of the polynomial used for smoothing of $Z_q$ for $q=1, 2, \cdots, Q$.  Note that Nadaraya-Watson estimate \citep{MR0185765} is a trivial case of the polynomial smoothing where the degree of the polynomial is zero. Suppose $K(\cdot)$ is the kernel function, and denote $K_{h_q}(z) = h_q^{-1}K(\frac{z}{h_q})$, where $h_q$ is the bandwidth parameter. Then, the smoothing matrix of $Z_q$ is given by
\be
\bo{S}_{P+q} = \left( \Z_{z_q}^{T} \bo{K}_{z_q} \Z_{z_q} \right)^{-1} \Z_{z_q}^{T} \bo{K}_{z_q}, \mbox{ for } q = 1, 2, \cdots, Q,
\label{sq}
\ee
where $\bo{K}_{z_q} = {\rm{diag}}\{ K_{h_q}(Z_{q1} - z_q), \cdots , K_{h_q}(Z_{qn} - z_q)\}$ is a diagonal matrix containing the kernel weight, and 
\be
\Z_{z_q} = 
\begin{bmatrix}
    1 & (Z_{q1} - z_q) & \dots  & (Z_{q1} - z_q)^{d_q} \\
    1 & (Z_{q2} - z_q) & \dots  & (Z_{q2} - z_q)^{d_q} \\
    \vdots & \vdots  & \ddots & \vdots \\
    1 & (Z_{qn} - z_q) & \dots  & (Z_{qn} - z_q)^{d_q}
\end{bmatrix}.
\label{z}
\ee
% Let us define $\bo{m}_p = (m_p(X_{p1}), \cdots, m_p(X_{pn}))^T$ for $p=1,2, \cdots, P$, and $\bo{m}_{P+q} = (m_{P+q}(Z_{q1}), \cdots, m_{P+q}(Z_{qn}))^T$ for $q=1,2, \cdots, Q$.
Then, the normal equations for the backfitting estimators  \citep{MR994249,MR1631333} are given by
\be
\begin{bmatrix}
    \bo{I}_n & \bo{S}_1^* & \dots  & \bo{S}_1^* \\
    \bo{S}_2^* & \bo{I}_n & \dots  & \bo{S}_2^* \\
    \vdots & \vdots  & \ddots & \vdots \\
    \bo{S}_{P+Q}^* & \bo{S}_{P+Q}^* & \dots  &  \bo{I}_n
\end{bmatrix}
\begin{bmatrix}
     \bo{m}_1 \\
    \bo{m}_2 \\
    \vdots \\
     \bo{m}_{P+Q}
\end{bmatrix}
=
\begin{bmatrix}
     \bo{S}_1^* \\
    \bo{S}_2^* \\
    \vdots \\
     \bo{S}_{P+Q}^*
\end{bmatrix} \bo{Y}^*,
\label{normal_eqn}
\ee
where $\bo{S}_d^* = (\bo{I}_n - \bo{1}_n \bo{1}_n^T/n)\bo{S}_d$ is the centered smoothing matrix for $d=1,2, \cdots, P+Q$, $\bo{Y}^* = \bo{Y} - \bar{Y}\bo{1}_n$  and $\bo{1}_n$ is the $n$-dimensional vector of elements 1. The solution to the normal equations has the form
\be
\begin{bmatrix}
     \widehat{\bo{m}}_1 \\
    \widehat{\bo{m}}_2 \\
    \vdots \\
     \widehat{\bo{m}}_{P+Q}
\end{bmatrix}
=
\begin{bmatrix}
    \bo{I}_n & \bo{S}_1^* & \dots  & \bo{S}_1^* \\
    \bo{S}_2^* & \bo{I}_n & \dots  & \bo{S}_2^* \\
    \vdots & \vdots  & \ddots & \vdots \\
    \bo{S}_{P+Q}^* & \bo{S}_{P+Q}^* & \dots  &  \bo{I}_n
\end{bmatrix}^{-1}
\begin{bmatrix}
     \bo{S}_1^* \\
    \bo{S}_2^* \\
    \vdots \\
     \bo{S}_{P+Q}^*
\end{bmatrix} \bo{Y}^* \equiv \bo{M}^{-1} \bo{C} \bo{Y}^*,
\label{normal_eqn_soln}
\ee
provided  the inverse exists. Here, $\bo{M}$ and $\bo{C}$ are the associated matrices. So, the backfitting estimator of $\bo{m}_d$ is given by
\be
\widehat{\bo{m}}_d = \bo{E}_d \bo{M}^{-1} \bo{C} \bo{Y}^* \equiv \bo{W}_d \bo{Y}^* , \ d = 1,2, \cdots, P+Q,
\label{w_d}
\ee
where $\bo{W}_d = \bo{E}_d \bo{M}^{-1} \bo{C}$, and  $\bo{E}_d$ is a block matrix of dimension $n \times n(P+Q)$ with $n \times n$ identity matrix in the $d$-th block and zero elsewhere.

Let us denote $\bo{W} = \sum_{d=1}^{P+Q} \bo{W}_d$. Suppose $\bo{W}^{[-d]}$ is the smoother matrix for the additive model after dropping out the term containing $\bo{m}_d, \ d = 1,2, \cdots, P+Q$. Then, the following lemma from \cite{MR1763322} ensures the existence and uniqueness of the backfitting estimators of the additive model.

\begin{lemma}
 If $|| \bo{S}_d^* \bo{W}^{[-d]} || <1$ for some $d \in (1,2, \cdots , P+Q)$, where $|| \cdot ||$ denotes any matrix norm, then the backfitting estimators uniquely exist and
%  \ba
%  \bo{W}_d &=& \bo{I}_n - \left(\bo{I}_n - \bo{S}_d^* \bo{W}^{[-d]} \right)^{-1}  \left(\bo{I}_n -  \bo{S}_d^* \right)\nn
%  &=&  \left(\bo{I}_n - \bo{S}_d^* \bo{W}^{[-d]} \right)^{-1} \bo{S}_d^* \left(\bo{I}_n -  \bo{W}^{[-d]} \right).
%  \label{wd}
%  \ea
 \be
 \bo{W}_d = \bo{I}_n - \left(\bo{I}_n - \bo{S}_d^* \bo{W}^{[-d]} \right)^{-1}  \left(\bo{I}_n -  \bo{S}_d^* \right)
 =  \left(\bo{I}_n - \bo{S}_d^* \bo{W}^{[-d]} \right)^{-1} \bo{S}_d^* \left(\bo{I}_n -  \bo{W}^{[-d]} \right).
 \label{wd}
 \ee
\end{lemma}

\section*{Appendix C: Proofs}

%\section{Appendix C: Proofs}

\begin{lemma}
 Let us assume that conditions (C2)--(C5) hold, then 
 \be
 \bo{S}_d^*  = \bo{S}_d - \frac{\bo{1}_n \bo{1}_n^T}{n} + o\left(\frac{\bo{1}_n \bo{1}_n^T}{n}\right) \mbox{ a.s.},
 \ee
 for all $d=1,2, \cdots , P+Q$.  The term $o\left(\frac{\bo{1}_n \bo{1}_n^T}{n}\right)$ means that each element is of order $o\left(\frac{1}{n}\right)$.
 \label{lemma:s}
\end{lemma}

\begin{proof}
 This property is proved in \cite{MR1429922} using a local polynomial fitting where the smoothing matrix is as defined in (\ref{sq}). This is also true for the point mass bin smoother as 
 \be
 \bo{S}_d^* = \bo{S}_d - \frac{\bo{1}_n \bo{1}_n^T}{n} \bo{S}_d 
 = \bo{S}_d - \frac{\bo{1}_n \bo{1}_n^T}{n} \mbox{ for } d=1,2, \cdots , P.
 \label{s1}
 \ee
 Note that for $d=1,2, \cdots , P$ the relationship is exact, and we do not need any assumption for this. 
\end{proof}

\begin{lemma}
If the predictors and covariates are pairwise independent then, under conditions (C1)--(C6),  we have
 \be
 \bo{S}_d^* \bo{S}_{d'}^*  = o\left(\frac{\bo{1}_n \bo{1}_n^T}{n}\right) \mbox{ a.s.},
 \ee
 for all $d \neq d' \in \{1,2, \cdots , P+Q\}$.
 \label{lemma:ss0}
\end{lemma}

\begin{proof}
Using equation (\ref{s1}) for $p$ and $p' \in \{ 1,2, \cdots , P\}$, we get 
 \ba
 \bo{S}_p^* \bo{S}_{p'}^* &=& \left(\bo{S}_p - \frac{\bo{1}_n \bo{1}_n^T}{n}\right) \left(\bo{S}_{p'} - \frac{\bo{1}_n \bo{1}_n^T}{n}\right) \nn
 &=& \bo{S}_p \bo{S}_{p'} - \frac{\bo{1}_n \bo{1}_n^T}{n} \bo{S}_p  - \frac{\bo{1}_n \bo{1}_n^T}{n} \bo{S}_{p'} + \frac{\bo{1}_n \bo{1}_n^T}{n}  \nn
 &=& \bo{S}_p \bo{S}_{p'} -  \frac{\bo{1}_n \bo{1}_n^T}{n}.
 \label{s*}
  \ea  
   Note that $\bo{S}_p \bo{S}_p = \bo{S}_p$ for $p=1,2, \cdots , P$. 
 For $p \neq p' \in \{1,2, \cdots , P\}$, we define $\bo{U} =  \bo{S}_p \bo{S}_{p'}$. Here $\bo{U}$ is a block matrix containing each element in the $rs$-th block equal to
 \be
 u_{rs} = \sum_{i:X_{pi}=x_{pr}, X_{p'i}=x_{p's}} \frac{1}{n_{pr} n_{p's}}, 
 \ee
 where $r  = 1, 2, \cdots , k_p$ and $s  = 1, 2, \cdots , k_{p'}$.
 Using strong law of large numbers (SLLN) and  assumption (C1), we get 
 \ba
 n u_{rs} &\overset{\mbox{a.s.}}{\longrightarrow}& \frac{1}{c_{pr} c_{p's}} P( X_p = x_{pr}, X_{p'} = x_{p's} )  \nn
  &=& \frac{P( X_p = x_{pr}, X_{p'} = x_{p's} )  }{P( X_p = x_{pr}) P( X_{p'} = x_{p's} )} .
  \label{nu}
 \ea
 Combining equations (\ref{s*}) and (\ref{nu}) the $ij$-th element of 
 $\bo{S}_p^* \bo{S}_{p'}^*$ becomes
 \be
 (\bo{S}_p^* \bo{S}_{p'}^*)_{ij} = \frac{1}{n}\left( \frac{P( X_p = x_{pi}, X_{p'} = x_{p'j} )  }{P( X_p = x_{pi}) P( X_{p'} = x_{p'j} )}  - 1 \right) \mbox{ a.s.}
 \ee
  So,  the lemma is proved for $p \neq p' \in \{1,2, \cdots , P\}$.  For $d \neq  d' \in \{P+1,P+2, \cdots , P+Q\}$ \cite{MR1429922} have shown that under conditions (C2)--(C6)
   \be
 (\bo{S}_d^* \bo{S}_{d'}^*)_{ij} = \frac{1}{n}\left( \frac{f_{dd'}( z_{di}, z_{d'j} )  }{f_d(z_{di}) f_{d'}(z_{d'j} )}  - 1 \right) + o\left(\frac{\bo{1}_n \bo{1}_n^T}{n}\right) \mbox{ a.s.},
 \ee
 where $f_{dd'}(\cdot)$ is the joint distribution of $Z_d$ and $Z_{d'}$, whereas $f_d(\cdot)$ and $f_{d'}(\cdot)$ are their marginal distributions. So, the lemma is true for  $d \neq  d' \in \{P+1,P+2, \cdots , P+Q\}$. Now, using condition (C8) and applying the same technique, we can prove this result when $d =1,2, \cdots , P$, and  $d' =P+1,P+2, \cdots , P+Q$, or vise versa.
  \end{proof}
  
% \begin{lemma}
%  Let us assume that conditions (C1)--(C6) hold, then 
%  \be
%  ( \bo{I}_n - \bo{S}_d^* \bo{S}_{d'}^*)^{-1}  = \bo{I}_n + O\left(\frac{\bo{1}_n \bo{1}_n^T}{n}\right) \mbox{ a.s.},
%  \ee
%  for all $d \neq d' =1,2, \cdots , P+Q$.
%  \label{lemma:ss}
% \end{lemma}

% \begin{proof}
% % Using equation (\ref{s1}) for $d, {d'} = 1,2, \cdots , P$, we get 
% %  \ba
% %  \bo{S}_d^* \bo{S}_{d'}^* &=& \left(\bo{S}_d - \frac{\bo{1}_n \bo{1}_n^T}{n}\right) \left(\bo{S}_{d'} - \frac{\bo{1}_n \bo{1}_n^T}{n}\right) \nn
% %  &=& \bo{S}_d \bo{S}_{d'} - \frac{\bo{1}_n \bo{1}_n^T}{n} \bo{S}_d  - \frac{\bo{1}_n \bo{1}_n^T}{n} \bo{S}_{d'} + \frac{\bo{1}_n \bo{1}_n^T}{n}  \nn
% %  &=& \bo{S}_d \bo{S}_{d'} -  \frac{\bo{1}_n \bo{1}_n^T}{n}.
% %  \label{s*}
% %   \ea
% %  Now $\bo{S}_d \bo{S}_d = \bo{S}_d$ for $d=1,2, \cdots , P$. Hence, comparing element-wise, we find
% %  $\bo{S}_d \bo{S}_{d'} \leq \min\{\bo{S}_d, \bo{S}_{d'}\}$ for $d \neq  {d'} =1,2, \cdots , P$, and from equation (\ref{s*}) we get 
% %  \be
% %  \bo{S}_d^* \bo{S}_{d'}^* = O\left(\frac{\bo{1}_n \bo{1}_n^T}{n}\right) \mbox{a.s., for } d \neq  d' =1,2, \cdots , P.
% %  \label{ss'}
% %  \ee
% %  Hence the lemma is proved for $d \neq  d' =1,2, \cdots , P$.
% %  For $p \neq p' =1,2, \cdots , P$ the lemma follows from Lemma \ref{lemma:ss0}. 
% %  For $d \neq  d' =P+1,P+2, \cdots , P+Q$, the result is shown in \cite{MR1429922}. The same technique will be used to prove this result when $d =1,2, \cdots , P$, and  $d' =P+1,P+2, \cdots , P+Q$, or vise versa.
% The lemma is an immediate consequence of Lemma \ref{lemma:ss0}.
% \end{proof}

\begin{lemma}
Let us denote $\bo{W} = \sum_{d=1}^{P+Q} \bo{W}_d$, where $\bo{W}_d$ is given in equation (\ref{w_d}). Then, under conditions (C1)--(C6) 
\be
 \bo{W}  \approx  \bo{S}^* + o\left(\frac{\bo{1}_n \bo{1}_n^T}{n}\right) \mbox{ a.s.},
 \ee
 where $\bo{S}^* = \sum_{d=1}^{P+Q} \bo{S}_d^*$.
 \label{lemma:W}
\end{lemma}

\begin{proof}
 For $P+Q=2$, we get
 \be
 \bo{W}^{[-1]} = \bo{S}_2^* \mbox{ and } \bo{W}^{[-2]} = \bo{S}_1^*.
 \ee
From equation (\ref{wd}), we have
 \be
 \bo{W}_1 = \left(\bo{I}_n - \bo{S}_1^* \bo{S}_2^* \right)^{-1} \bo{S}_1^* \left(\bo{I}_n -  \bo{S}_2^* \right) .
 \label{w1}
 \ee
 Using Lemma \ref{lemma:ss0}, we have the following approximation
  \be
 ( \bo{I}_n - \bo{S}_1^* \bo{S}_2^*)^{-1}  \approx \bo{I}_n + o\left(\frac{\bo{1}_n \bo{1}_n^T}{n}\right) \mbox{ a.s.}
 \label{is}
 \ee
 This approximation is exact when the corresponding predictors or covariates are pairwise independent. Combining equations (\ref{w1}) and (\ref{is}), we get
 \be
 \bo{W}_1 \approx \bo{S}_1^* + o\left(\frac{\bo{1}_n \bo{1}_n^T}{n}\right) \mbox{ a.s}.
 \label{w11}
 \ee
 Similarly $\bo{W}_2  \approx \bo{S}_2^* + o\left(\frac{\bo{1}_n \bo{1}_n^T}{n}\right) \mbox{ a.s}.$ 
 Now, for all values of $P$ and $Q$, we prove by recursion that
 \be
 \bo{W}_d \approx \bo{S}_d^* + o\left(\frac{\bo{1}_n \bo{1}_n^T}{n}\right) \mbox{ a.s}.
 \ee
Therefore, by taking summation over $d =1,2, \cdots , P+Q$ the lemma is proved.
\end{proof}

% \begin{lemma}
% Under conditions (C1)--(C6) 
% \be
%  \left(\bo{I}_n - \bo{S}_d^* \bo{W}^{[-d]} \right)^{-1}=  \bo{I}_n + O\left(\frac{\bo{1}_n \bo{1}_n^T}{n}\right) \mbox{ a.s. },
%  \ee
%  for all $ d =1,2, \cdots , P+Q.$
%  \label{lemma:SW}
% \end{lemma}
% 
% \begin{proof}
%  By the same logic we can show that
%  \be
%  \bo{W}^{[-d]} =  \bo{S}^{[-d]} + O\left(\frac{\bo{1}_n \bo{1}_n^T}{n}\right) \mbox{ a.s., for } d =1,2, \cdots , P+Q,
%  \ee
%  where
%  $\bo{S}^{[-d']} = \bo{S} - \bo{S}_{d}$. Now
%  \ba
%  \bo{S}_d^* \bo{W}^{[-d]} &=& \sum_{d'=1, d' \neq d}^{P+Q} \bo{S}_d^* \bo{S}_{d'} + O\left(\frac{\bo{1}_n \bo{1}_n^T}{n}\right) \mbox{ a.s.}\nn
%  &=& \sum_{d'=1, d' \neq d}^{P+Q} \bo{S}_d^* \bo{S}_{d'}^* + O\left(\frac{\bo{1}_n \bo{1}_n^T}{n}\right) \mbox{ a.s.},
%  \ea
%  by using Lemma \ref{lemma:s}. From equation (\ref{ss'}) we get 
%   \be
%  \bo{S}_d^* \bo{W}^{[-d]} = O\left(\frac{\bo{1}_n \bo{1}_n^T}{n}\right) \mbox{ a.s.}.
%  \ee
%  Hence the lemma is proved.
% \end{proof}

\begin{lemma}
Suppose conditions (C1)--(C4) and (C8) are satisfied. Then, under $H_0^{**}$,
\be
  \bo{m}^T \bo{S}_p^* \bo{m} =   o_p(1),
 \ee
 for all $ p =1,2, \cdots , P$, where $\bo{m} = \sum_{p=1}^{P+Q} \bo{m}_p$.
 \label{lemma:Sm}
\end{lemma}

\begin{proof}
 Note that 
 \be
 \bo{S}_p \bo{m}_p = \bo{m}_p \mbox{ for all } p =1,2, \cdots , P.
 \label{smp}
 \ee
 For $ p =1,2, \cdots , P$ and $q = 1,2, \cdots , Q$ we define $\bo{u} =  \bo{S}_p \bo{m}_{P+q} = (u_1 \bo{1}_{n_{p1}}^T, u_2 \bo{1}_{n_{p2}}^T,$ $ \cdots, u_{n_{pk_p}} \bo{1}_{n_{pk_p}}^T)^T$. Then 
 \be
 u_j = \sum_{i:X_{pi}=x_{pj}} \frac{m_{P+q} (Z_{qi})}{n_{pj}}, \mbox{ for } j  = 1, 2, \cdots , k_p.
 \ee
 Using strong law of large numbers (SLLN) and assumption (C8) we get 
 \be
 u_j \overset{\mbox{a.s.}}{\longrightarrow}  E[ m_{P+q} (Z_q)| X_p = x_{pj} ] = 0 .
 \ee
 So 
 \be
 \bo{S}_p \bo{m}_{P+q} = \bo{O}_{n,1} \mbox{ a.s. for all }  p =1,2, \cdots , P \mbox{ and } q  = 1,2, \cdots , Q.
 \label{sm0}
 \ee
Hence, under $H_0^{**}$,  for all $p =1,2, \cdots , P$
 \be
  \bo{m}^T \bo{S}_p \bo{m} = \left(\sum_{d=1}^{P} \bo{m}_p^T \right) \bo{S}_p \left(\sum_{d=1}^{P} \bo{m}_p \right) + o_p(1)
  = o_p(1) .
  \label{msm}
 \ee
 Again using SLLN we get 
 \be
 \frac{\bo{1}_n^T \bo{m}_{P+q}}{n} \overset{\mbox{a.s.}}{\longrightarrow}  E[ m_{P+q} (Z_q)] = 0
 \ee
 for all $q  = 1,2, \cdots , Q$. Similarly, $\bo{1}_n^T \bo{m}_p/n = 0$ a.s. for all $p  = 1,2, \cdots , P$. Hence using Lemma \ref{lemma:s} the lemma is proved from equation (\ref{msm}).
%  Using strong law of large numbers, we get 
%  \be
%  u_j \rightarrow E[ m_{p'} (X_{p'})| X_p = j ]  \mbox{ a.s}.
%  \ee
%  So
%  \be
%  m_p(X_p = j) u_j \rightarrow E[ m_{p'} (X_{p'}) ]  \mbox{ a.s}.
%  \ee
%  Now from the assumption of the additive model (\ref{model0}), we have $E[ m_{p'} (X_{p'})]=0$. So 
%  \be
%  m_p \bo{S}_p m_{p'} = 0 \mbox{ a.s. for all } p' \neq p = 1,2, \cdots , P.
%  \ee
%  Similarly 
%  \be
%  m_p \bo{S}_p m_q = 0 \mbox{ a.s. for all } p = 1,2, \cdots , P, \mbox{ and } q = 1,2, \cdots , Q.
%  \label{smq}
%  \ee
%  Combining (\ref{smp}) and (\ref{smq}), we get 
%  \be
%   \bo{m}^T \bo{S}_p \bo{m} = \sum_{d=1}^{P+Q} m_d \bo{S}_p m_{d'}
%   = \sum_{i=1}^n m_p^2(X_{pi}) \mbox{ a.s}.
%  \ee
%  As $\bo{1}^T \bo{m}/n \rightarrow 0$ a.s., we prove the lemma using Lemma \ref{lemma:s}.
\end{proof}

\begin{lemma}
Denote $\bo{A}_{2n} = ( \bo{W} - \bo{I}_n )^T ( \bo{W} - \bo{I}_n )$, then under conditions (C1)--(C6) 
 \be
  \bo{A}_{2n} \approx \bo{S}^{*T} \bo{S}^* -\bo{S}^* - \bo{S}^{*T} + \bo{I}_n + o\left(\frac{\bo{1}_n \bo{1}_n^T}{n}\right) \mbox{ a.s.},
 \ee
 where $\bo{S}^* = \sum_{d=1}^{P+Q} \bo{S}^*_d$.
 \label{lemma:A}
\end{lemma}

\begin{proof}
Using Lemma \ref{lemma:W}, we get
  \ba
  \bo{A}_{2n} &\approx& ( \bo{S}^* - \bo{I}_n )^T ( \bo{S}^* - \bo{I}_n ) + o\left(\frac{\bo{1}_n \bo{1}_n^T}{n}\right) \mbox{ a.s.}\nn
  &=& \bo{S}^{*T} \bo{S}^* -\bo{S}^* - \bo{S}^{*T} + \bo{I}_n + o\left(\frac{\bo{1}_n \bo{1}_n^T}{n}\right) \mbox{ a.s}.
 \ea
\end{proof}

\begin{proof}[Corollary \ref{theorem:chi}]
 Denote $\bo{A}_{1n} = ( \bo{W}_{[Z]} - \bo{I}_n )^T ( \bo{W}_{[Z]} - \bo{I}_n )$, where $\bo{W}_{[Z]}$ is the smoother matrix for the additive model after dropping all $P$ predictors. Using an argument similar to that in the proof of Lemma \ref{lemma:A}, we  find
 \be
  \bo{A}_{1n} \approx \bo{S}_{[Z]}^{*T} \bo{S}_{[Z]}^* -\bo{S}_{[Z]}^* - \bo{S}_{[Z]}^{*T} + \bo{I}_n + o\left(\frac{\bo{1}_n \bo{1}_n^T}{n}\right) \mbox{ a.s.},
 \ee
 where
 $\bo{S}_{[Z]}^* =\sum_{d=1}^{Q} \bo{S}_{P+d}^*$. So
 \be
 \bo{A}_{1n} - \bo{A}_{2n} \approx \bo{S}_{[X]}^* + \bo{S}_{[X]}^{*T} - \bo{S}_{[X]}^{*T}\bo{S}_{[X]}^* - \bo{S}_{[X]}^{*T} \bo{S}_{[Z]}^* - \bo{S}_{[X]}^* \bo{S}_{[Z]}^{*T} + o\left(\frac{\bo{1}_n \bo{1}_n^T}{n}\right) \mbox{ a.s.},
 \label{a1a2}
 \ee
$\bo{S}_{[X]}^* =\sum_{d=1}^{P} \bo{S}_{d}^*$. Now $\bo{S}_d^* \bo{S}_d^* = \bo{S}_d^*$ and $\bo{S}_d^{*T}=\bo{S}_d^*$ for $d=1,2, \cdots , P$. From the technique used in Lemma 3, it can be shown that $\bo{S}_{[X]}^* \bo{S}_{[Z]}^{*T}$ is a symmetric matrix. So, using Lemma \ref{lemma:ss0} equation (\ref{a1a2})  reduces to
 \be
 \bo{A}_{1n} - \bo{A}_{2n} \approx \bo{S}_{[X]}^* + o\left(\frac{\bo{1}_n \bo{1}_n^T}{n}\right) 
 = \sum_{d=1}^{P} \bo{S}_{d}^* + o\left(\frac{\bo{1}_n \bo{1}_n^T}{n}\right) \mbox{ a.s}.
 \ee
Now
\ba
RSS_0^{**} - RSS_1^* &=& \bo{Y}^{*T} (\bo{A}_{1n} - \bo{A}_{2n})  \bo{Y}^* \nn
% &=& \sum_{d=1}^{P} \bo{Y}^{*T} \bo{S}_{d}  \bo{Y}^* + \bo{Y}^{*T} O\left(\frac{ \bo{1}_n \bo{1}_n^T }{n}\right)\bo{Y}^*  \nn
&\approx& \sum_{d=1}^{P} \bo{Y}^{*T} \bo{S}_{d}^*  \bo{Y}^* +  o_p(1) \nn
&=& \sum_{d=1}^{P} \bo{Y}^T \left(\bo{I}_n - \frac{ \bo{1}_n \bo{1}_n^T }{n}\right)^T \bo{S}_{d}^* \left(\bo{I}_n - \frac{ \bo{1}_n \bo{1}_n^T }{n}\right) \bo{Y} +  o_p(1) \nn
% &=& \sum_{d=1}^{P} \bo{Y}^T \left(\bo{S}_{d} - \frac{ \bo{1}_n \bo{1}_n^T }{n}\right)   \bo{Y} +  o_p(1)\nn
&=& \sum_{d=1}^{P} \bo{Y}^T \bo{S}_{d}^*    \bo{Y} +  o_p(1)
\label{rss012}
\ea
Let us define $\bo{\epsilon} =(\epsilon_1, \cdots, \epsilon_n)^T$. Then
\be
RSS_0^{**} - RSS_1^* =  \sum_{d=1}^{P} \bo{m}^T \bo{S}_d^*   \bo{m} 
+ 2\sum_{d=1}^{P} \bo{\epsilon}^T \bo{S}_d^*   \bo{m} 
+ \sum_{d=1}^{P} \bo{\epsilon}^T \bo{S}_{d}^*    \bo{\epsilon} +  o_p(1).
\label{rss018}
\ee
Using Lemma \ref{lemma:Sm},  under $H_0^{**}$, we have
\be
\bo{m}^T \bo{S}_d^*   \bo{m} = o_p(1) \mbox{ for all } d = 1, 2, \cdots, P.
\ee
Moreover, using (\ref{sm0}) it is easy to show that, under $H_0^{**}$, $\bo{\epsilon}^T \bo{S}_d^*   \bo{m} = o_p(1)$ for all $d = 1, 2, \cdots, P$. Hence, from (\ref{rss018}), we get
\be
RSS_0^{**} - RSS_1^* =   \sum_{d=1}^{P} \bo{\epsilon}^T \bo{S}_{d}^*    \bo{\epsilon} +  o_p(1).
\label{rss01}
\ee
Now, for $d=1,2, \cdots , P$, using the definition in (\ref{sp}), we have
\be
\bo{\epsilon}^T \bo{S}_{d}  \bo{\epsilon} = \sum_{j=1}^{k_d} \frac{1}{n_{dj}}(\bo{e}_{dj}^T\bo{\epsilon})^T   \bo{e}_{dj}^T\bo{\epsilon},
\label{ese}
\ee
where $\bo{e}_{dj}$ is a vector with $n_{dj}$ elements one and rest are zero. If $X_{di}=x_{dj}$, then the $i$-th element of $\bo{e}_{dj}$ is one. Note that in this definition, we did not sort $X_d$ according to their observed values. As $E[\epsilon_i^2 ] < \infty$ under condition (C7), using central limit theorem (CLT), we get 
\be
 \frac{1}{\sigma \sqrt{n_{dj}}} \bo{e}_{dj}^T\bo{\epsilon}
 \overset{a}{\equiv} U_{dj} \sim N(0,1).
\ee
As $\bo{e}_{dj}^T\bo{\epsilon}$ and $\bo{e}_{dj'}^T\bo{\epsilon}$ are independent for all $j\neq j' \in \{1, 2, \cdots, k_d\}$, the components of $\bo{U}_d = (U_{d1}, U_{d2}, \cdots, U_{dk_d})^T$ are  i.i.d. standard normal variables. Therefore, from (\ref{ese}), we have $\frac{1}{\sigma^2}\bo{\epsilon}^T \bo{S}_d  \bo{\epsilon}
\overset{a}{\sim} \chi^2(k_d) .$ 
Let us define 
\be
\bar{U} = \sum_{j=1}^{k_d} \sqrt{\frac{n_{dj}}{n}} U_{dj}.
\ee
Then $\bar{U} = c_d^T \bo{U}_d$ a.s., where $c_d = (\sqrt{c_{d1}}, \sqrt{c_{d2}}, \cdots, \sqrt{c_{dk_d}})^T$. Note that $\bar{U}^2 = \frac{1}{n} \bo{\epsilon}^T \bo{\epsilon}$. 
Hence 
\be
\frac{1}{\sigma^2}\bo{\epsilon}^T \bo{S}_{d}^*  \bo{\epsilon}
\overset{a}{\equiv} \bo{U}_d^T \left(\bo{I}_{k_d} - c_d c_d^T \right) \bo{U}_d 
\sim \chi^2(k_d-1) ,
\label{chi}
\ee
because $(\bo{I}_{k_d} - c_d c_d^T)$ is an idempotent matrix of rank $(k_d-1)$.
% 
% Let us define 
% \be
% \X_{(p)} = 
% \begin{bmatrix}
%     1 & X_{p1} & X_{p1}^2 & \dots  & X_{p1}^{k_p-1} \\
%     1 & X_{p2} & X_{p2}^2 & \dots  & X_{p2}^{k_p-1} \\
%      \vdots  & \ddots & \vdots \\
%     1 & X_{pn} & X_{pn}^2 & \dots  & X_{pn}^{k_p-1} \\
% \end{bmatrix},
% \label{xps}
% \ee
% Then it can be shown that 
% \be
% \bo{S}_p^* =  \X_{(p)} \left(\X_{(p)}^T  \X_{(p)}\right)^{-1} \X_{(p)}^T - \frac{\bo{1}_n \bo{1}_n^T}{n} .
% \ee
% Now columns of the matrix $\X_{(p)} \left(\X_{(p)}^T  \X_{(p)}\right)^{-1} \X_{(p)}^T - \frac{\bo{1}_n \bo{1}_n^T}{n}$ form a orthogonal basis for the column space of $\X_{(p)}$ excluding the first column of ones. According to assumption (C9), all these column spaces for $p=1,2, \cdots, P$ are linearly independent. So, the rank of $\sum_{p=1}^P \bo{S}_p^*$ is $\sum_{p=1}^P (k_p -1)$. Moreover, as $\bo{S}_p^*$ is idempotent for all $p=1,2,\cdots,P$, using approximation of Lemma \ref{lemma:ss0} it can be shown that $\sum_{p=1}^P \bo{S}_p^*$ is also idempotent.
If all predictors are pairwise independent, then from Lemma \ref{lemma:ss0}, we get 
\be
\bo{S}_{d}^* \bo{S}_{d'}^*
%= \left(\bo{S}_{d} - \frac{ \bo{1}_n \bo{1}_n^T }{n}\right)  \left(\bo{S}_{d'} - \frac{ \bo{1}_n \bo{1}_n^T }{n}\right)
= o\left( \frac{\bo{1}_n \bo{1}_n^T}{n} \right) \mbox{ a.s.}
\ee
for $d \neq d' \in \{ 1, 2, \cdots, P\}$. So, $\bo{\epsilon}^T \bo{S}_{d}^*  \bo{\epsilon}$ and $\bo{\epsilon}^T \bo{S}_{d'}^*  \bo{\epsilon}$ are asymptotically independent for all $d \neq d'$  (see p. 84 of \citealp{bapat2012linear}). 
 Therefore, if the predictors are pairwise independent, then under $H_0^{**}$,  equation (\ref{rss01}) gives
\be
\frac{1}{\sigma^2}(RSS_0^{**} - RSS_1^*) 
\overset{a}{\sim} \chi^2\left(\sum_{d=1}^P (k_d-1) \right).
\label{rss1}
\ee
 However, in general, the above distribution  comes out to be a sum of $P$ dependent chi-square variables as
 \be
\frac{1}{\sigma^2}(RSS_0^{**} - RSS_1^*) 
\overset{a}{\equiv}  \sum_{d=1}^P \bo{U}_d^T \left(\bo{I}_{k_d} - c_d c_d^T \right) \bo{U}_d.
\label{rss1g}
\ee
Suppose $d\neq d' \in \{1, 2, \cdots, P\}$, $j=1,2,\cdots, k_d$ and $j'=1,2,\cdots, k_{d'}$, then the correlation between $U_{dj}$ and $U_{d'j'}$ is given by
\ba
 {\rm Corr}(U_{dj}, U_{d'j'}) &=& \lim_{n \rightarrow \infty} {\rm Corr} \left( \frac{1}{\sigma \sqrt{n_{dj}}} \bo{e}_{dj}^T\bo{\epsilon}, \frac{1}{\sigma \sqrt{n_{d'j'}}} \bo{e}_{d'j'}^T\bo{\epsilon} \right) \nn
 &=& \lim_{n \rightarrow \infty} \sum_{i:X_{di}=x_{dj}, X_{d'i}=x_{d'j'}} \frac{1}{\sqrt{n_{dj} n_{d'j'}}} \nn
 &=& \frac{1}{\sqrt{c_{dj} c_{d'j'}}} P( X_d = x_{dj}, X_{d'} = x_{d'j'} ).
 \label{cc}
\ea
The last expression is derived using the same techniques as used in equation (\ref{nu}). Note that combining equations (\ref{rss1g}) and (\ref{cc}), we obtain result (\ref{rss1}) if the predictor variables are independent. 
% However, in general, the above distribution will come out to be a linear combination of the chi-square variables. Let us define $\bo{S}^* = \sum_{p=1}^P \bo{S}_p^*$. Now $\bo{S}^*$ can be diagonalized as $\bo{S}^* = \bo{R} \bo{\Lambda} \bo{R}'$, where $\bo{\Lambda} = \rm{diag}(\xi_1, \xi_2, \cdots, \xi_k)$ is the diagonal matrix of non-zero eigenvalues of $\bo{S}^*$,  $\bo{R}$ is the matrix of corresponding orthonormal eigenvectors, and $k$ is the rank of $\bo{S}^*$. Then, under $H_0$, 
% \be
% \frac{1}{\sigma^2}(RSS_0 - RSS_1) = \frac{1}{\sigma^2}\bo{\epsilon}^T \bo{S}^*  \bo{\epsilon}
% = \frac{1}{\sigma^2}\bo{\epsilon}^T \bo{R} \bo{\Lambda} \bo{R}'  \bo{\epsilon}
% = \frac{1}{\sigma^2} \sum_{i=1}^k \xi_i U_i^2,
% \label{rss1g}
% \ee
% where $U_i = \bo{R}_{\cdot i}'  \bo{\epsilon}$ with $\bo{R}_{\cdot i}$ being the $i$-th column of $\bo{R}$. Using CLT it can be shown that $U_1, U_2, \cdots, U_k$ are independent standard normal variables.

\noindent
Now, it is easy to show that
 \ba
 \frac{1}{n} RSS_1^* &=& \frac{1}{n}  \bo{Y}^{*T} \bo{A}_{2n} \bo{Y}^* \nn
 &\approx& \frac{1}{n} \bo{Y}^{*T} \left( \bo{S}^{*T} \bo{S}^* -\bo{S}^* - \bo{S}^{*T}  + \bo{I}_n  + o\left(\frac{\bo{1}_n \bo{1}_n^T}{n}\right)\right)\bo{Y}^*  \nn
%  &=& \frac{1}{n} \bo{Y}^{*T} \bo{Y}^* + \bo{Y}^{*T} O\left(\frac{ \bo{1}_n \bo{1}_n^T }{n^2}\right)\bo{Y}^* \nn
 &=& \sigma^2 + o_p(1) .
 \label{rss0}
 \ea
 Therefore, using Slutsky's theorem, $\sigma^2$ in (\ref{rss1g}) may be replaced by $\frac{1}{n} RSS_1^*$, and therefore %thus the proof is complete. 
 \be
\lambda_n (H_0^{**}) \overset{a}{\equiv} 
\sum_{p=1}^P \bo{U}_p^T \left(\bo{I}_{k_p} - c_p c_p^T \right) \bo{U}_p.
%= \sum_{p=1}^P \bo{U}_p^{*T} \bo{\Sigma}_1 \bo{U}_p^*,
\label{lm}
\ee
Define $\bo{U} = (\bo{U}_1^T, \bo{U}_2^T, \cdots, \bo{U}_P^T)^T$, and $\bo{U}^* = \bo{\Sigma}_1 \bo{U}$, where $\bo{\Sigma}_1$ is defined in Section \ref{sec:GLR}. As $\bo{\Sigma}_1$ is an idempotent matrix, $\lambda_n (H_0^{**})$ in equation (\ref{lm}) is written as
 \be
\lambda_n (H_0^{**}) \overset{a}{\equiv} 
 \bo{U}^T \bo{\Sigma}_1 \bo{U}
=  \bo{U}^{*T} \bo{U}^*,
\label{lm1}
\ee
Now, the covariance matrix of $\bo{U}$ is $\bo{\Sigma}_2$ (defined in Section \ref{sec:GLR}), which is a block matrix with $p$-th diagonal block is an identity matrix of order $k_p$, and the $ij$-th element of the $pp'$-th off-diagonal block is given in (\ref{cc}). So, the covariance matrix of $\bo{U}^*$ is $\bo{\Sigma}_1 \bo{\Sigma}_2 \bo{\Sigma}_1$. Let $\lambda_1, \lambda_2, \cdots, \lambda_s$ are non-zero eigenvalues of $\bo{\Sigma}_1\bo{\Sigma}_2 \bo{\Sigma}_1$, where $s$ is the rank of  $\bo{\Sigma}_1\bo{\Sigma}_2 \bo{\Sigma}_1$. Suppose $V=(V_1, V_2, \cdots, V_s)^T$ is a vector of i.i.d. standard normal variables, and $\Lambda = {\rm diag}(\lambda_1, \lambda_2, \cdots, \lambda_s)$. Then, from (\ref{lm1}), the theorem is proved as
 \be
\lambda_n (H_0^{**}) \overset{a}{\equiv} 
 \bo{V}^T \bo{\Lambda} \bo{V}
=  \sum_{i=1}^s \lambda_i V_i^2.
\label{lm2}
\ee

% Combining (\ref{rss1g}) and (\ref{rss0}) and using Slutsky's theorem the null distribution of the GLR test is approximated as
% \be
% \frac{n (RSS_0 - RSS_1)}{RSS_0} 
% \overset{a}{\sim} U .
% \ee
\end{proof}

\begin{theorem}
	Let us consider the notations and assumptions of Corollary \ref{theorem:chi}. Then, under $H_1$, the asymptotic distribution of the GLR test statistic coincides with $\delta^2 + \sum_{i=1}^s \lambda_i V_i^2$, where  $\delta^2 =  \sum_{r,s=1}^{P} E(m_r^*  m_s^*)$.
	% 
	% Under conditions (C1)--(C8) the asymptotic distribution of the GLR test statistic, under $H_1$, coincides with $\sum_{i=1}^k \xi_i (U_i + \bo{R}_{\cdot i}^T \bo{m}^*)^2$, where $U_1, U_2, \cdots, U_k$ are independent standard normal variables.
	\label{theorem:power}
\end{theorem}

\begin{proof}%[Proof of Theorem \ref{theorem:power}]
If $H_0^{**}$ is not true, then from (\ref{rss018}), we get 
\be
RSS_0^{**} - RSS_1^* =  \sum_{d=1}^{P} \bo{m}^T \bo{S}_d^*   \bo{m} 
+ \sum_{d=1}^{P} \bo{\epsilon}^T \bo{S}_{d}^*    \bo{\epsilon} +  o_p(1).
\label{rss1g20}
\ee
Now, the result in equation (\ref{rss1g}) reduces to
 \be
\frac{1}{\sigma^2}\sum_{d=1}^{P} \bo{\epsilon}^T \bo{S}_{d}^*    \bo{\epsilon}
=  \sum_{d=1}^P \bo{U}_d^T \left(\bo{I}_{k_d} - c_d c_d^T \right) \bo{U}_d.
\label{rss1g45}
\ee
% Using equation (\ref{sm0}) and putting the functional value of $m_p^*$, for $p=1,2, \cdots, P$, we get
% \ba
% \frac{1}{\sigma^2}(RSS_0 - RSS_1) &=& \frac{1}{\sigma^2}(\bo{m}^* + \bo{\epsilon})^T \bo{S}^*  (\bo{m}^* + \bo{\epsilon}) \nn
% &=& \frac{1}{\sigma^2} \sum_{i=1}^k \xi_i (U_i + \bo{R}_{\cdot i}^T\bo{m}^*)^2.
% \label{rss1g2}
% \ea
% Therefore, using (\ref{rss0}) and (\ref{rss1g2}) the theorem is prooved. 
% \end{proof}
% 
% \begin{proof}[Proof of Corollary \ref{corollary:power}]
% If the predictor variables are pairwise independent, then $\bo{S}^*$ becomes an idempotent matrix. Thus $\xi_i =1$ for all $i=1, 2, \cdots, k$, and $k=\sum_{p=1}^P (k_p - 1)$. So, the distribution of the GLR test statistic coincides with $\sum_{i=1}^k (U_i + \bo{R}_{\cdot i}'\bo{m}^*)^2$, which is a non-central chi-square with degree of freedom $k$ and the non-centrality parameter
% \be
% \delta = \sum_{i=1}^k ( \bo{R}_{\cdot i}'\bo{m}^*)^2 = \bo{m}^* \bo{R} \bo{\Lambda}  \bo{R}^T \bo{m}^{*T} =  \sum_{p=1}^{P}  \bo{m}^{*T} \bo{S}_p \bo{m}^*.
% \label{del}
% \ee
From the proof of Lemma \ref{lemma:Sm}, we have
 \be
  \bo{m}^T \bo{S}_p \bo{m} = \left(\sum_{d=1}^{P} \bo{m}_p^T \right) \bo{S}_p \left(\sum_{d=1}^{P} \bo{m}_p \right) + o_p(1).
  \label{msm1}
 \ee
Using $\bo{S}_p \bo{m}_p = \bo{m}_p$, we get $\bo{m}_r^T \bo{S}_p \bo{m}_p = \bo{m}_r^T \bo{m}_p$  for all $p, r =1,2, \cdots , P$. Hence
\be
\frac{1}{n}\bo{m}_r^T \bo{S}_p \bo{m}_p \overset{\mbox{a.s.}}{\longrightarrow}  E(\bo{m}_r^T \bo{m}_p).
\ee
Suppose $u_{ij}$ is the $(i,j)$-th element of $\bo{m}_r^T \bo{S}_p \bo{m}_s$ for some $p, r , s=1,2, \cdots , P$. Then
 \be
 u_{ij} = \sum_{l=1}^{k_p} \sum_{(i,j):X_{pi}=X_{pj}=x_{pl}} \frac{m_r (X_{ri}) m_s (X_{sj})}{n_{pl}}.
 \label{uij}
 \ee
 Note that
 \be
  \sum_{(i,j):X_{pi}=X_{pj}=x_{pl}} \frac{m_r (X_{ri}) m_s (X_{sj})}{n_{pl}} \overset{\mbox{a.s.}}{\longrightarrow}  E(m_r m_s | X_p = x_{pl}).
 \ee
 So, using condition (C1), we get from equation (\ref{uij})
 \be
  u_{ij} \overset{\mbox{a.s.}}{\longrightarrow}  E(m_r m_s) \mbox{ and } \frac{1}{n} u_{ij} \overset{\mbox{a.s.}}{\longrightarrow}  0.
 \ee
 Hence, from equation (\ref{msm1}), we get
  \be
  \frac{1}{n} \bo{m}^T \bo{S}_p \bo{m} = \sum_{r}^{P} E(m_r m_p) + o_p(1).
 \ee
 Therefore
  \be
  \frac{1}{n} \sum_{p=1}^{P}  \bo{m}^T \bo{S}_p \bo{m} = \sum_{r,s=1}^{P} E(m_r  m_s) + o_p(1).
  \label{msm2}
 \ee
 As $E(m_p)=0$ for all $p=1,2,\cdots,P$, we get
  \be
  \frac{1}{n} \sum_{p=1}^{P}  \bo{m}^T \bo{S}_p^* \bo{m} = \sum_{r,s=1}^{P} E(m_r  m_s) + o_p(1).
  \label{msm210}
 \ee
 Combining (\ref{rss0}), (\ref{rss1g20}), (\ref{rss1g45}) and (\ref{msm210}) the theorem is proved.
%  As the predictors are pairwise independent, combining (\ref{del}) and (\ref{msm2}) we get
%  \be
%  \delta =  \sum_{p=1}^{P} E(m_p^{*2}) + o_p(1).
%  \ee
\end{proof}

\begin{proof}[Corollary \ref{theorem:chi2}]
The residual sum of squares under $H_0^*$ can be written as
\ba
RSS_0^* &=& \left( \Y^* - \X^* \wtheta - \bo{W}_{[Z]} \left(\Y^* - \X^* \wtheta  \right) \right)^T \left( \Y^* - \X^* \wtheta - \bo{W}_{[Z]} \left(\Y^* - \X^* \wtheta  \right) \right) \nn
&=& \left( \Y^* - \left(\bo{I}_n - \bo{W}_{[Z]}\right) \X^* \wtheta - \bo{W}_{[Z]} \Y^*  \right)^T  \left( \Y^* - \left(\bo{I}_n - \bo{W}_{[Z]}\right) \X^* \wtheta - \bo{W}_{[Z]} \Y^*  \right)\nn
&=& \Y^{*T} \left(\bo{I}_n - \bo{A}_n - \bo{W}_{[Z]})^T (\bo{I}_n - \bo{A}_n - \bo{W}_{[Z]} \right) \Y^*,
\label{rss*}
\ea
where
\be
\bo{A}_n = \left(\bo{I}_n - \bo{W}_{[Z]}\right) \X^* \left(\X^{*T} \left(\bo{I}_n - \bo{W}_{[Z]}\right) \X^*\right)^{-1} \X^{*T} \left(\bo{I}_n - \bo{W}_{[Z]}\right).
\label{A}
\ee
 Using Lemma \ref{lemma:W} it can be shown that
 \be
 \bo{W}_{[Z]} = \sum_{q=1}^Q \bo{W}_{P+q} \approx  \sum_{q=1}^Q \bo{S}_{P+q}^* + o\left(\frac{\bo{1}_n \bo{1}_n^T}{n}\right) \mbox{ a.s.}
 \label{wz}
 \ee
 Hence
  \be
 \bo{I}_n - \bo{W}_{[Z]} \approx \bo{I}_n  + o\left(\frac{\bo{1}_n \bo{1}_n^T}{n}\right) \mbox{ a.s.}.
 \ee
 Therefore, equation (\ref{A}) reduces to
 \be
\bo{A}_n =  \X^* \left(\X^{*T}  \X^*\right)^{-1} \X^{*T} + o\left(\frac{\bo{1}_n \bo{1}_n^T}{n}\right) \mbox{ a.s.}
\label{A1}
\ee
% We have assumed that the limit $n^{-1} \X^{*T} \X^*$ exists, so
% \be
% \bo{A}_n \bo{W}_{[Z]} = O\left(\frac{\bo{1}_n \bo{1}_n^T}{n}\right) \mbox{ a.s.}
% \ee
As $\X^* \left(\X^{*T}  \X^*\right)^{-1} \X^{*T}$ is an idempotent matrix,  using (\ref{wz}) and (\ref{A1}) we get from equation (\ref{rss*}) 
% \begin{multline}
% RSS_0^* = \bo{Y}^{*T}  \Big(\bo{I}_n + \bo{W}_{[Z]}^T \bo{W}_{[Z]} - \bo{W}_{[Z]} - \bo{W}_{[Z]}^T - \X^* \left(\X^{*T}  \X^*\right)^{-1} \X^{*T} \\
% + \bo{W}_{[Z]}^T \X^* \left(\X^{*T}  \X^*\right)^{-1} \X^{*T} 
% +  \X^* \left(\X^{*T}  \X^*\right)^{-1} \X^{*T} \bo{W}_{[Z]}\Big) \bo{Y}^* + o_p(1).
% \end{multline}
% Using Lemma \ref{lemma:W} we get
\begin{multline}
RSS_0^* = \bo{Y}^{*T}  \Big(\bo{I}_n + \bo{S}_{[Z]}^{*T} \bo{S}_{[Z]}^* - \bo{S}_{[Z]}^* - \bo{S}_{[Z]}^{*T} - \X^* \left(\X^{*T}  \X^*\right)^{-1} \X^{*T} \\
+ \bo{S}_{[Z]}^{*T} \X^* \left(\X^{*T}  \X^*\right)^{-1} \X^{*T} 
+  \X^* \left(\X^{*T}  \X^*\right)^{-1} \X^{*T} \bo{S}_{[Z]}^*\Big) \bo{Y}^* + o_p(1).
\label{rss122*}
\end{multline}
Suppose $\btheta_0$ is the true value of $\btheta$ under $H_0^*$. So, under $H_0^*$, the model can be written as
\be
\bo{Y} = \alpha \bo{1}_n + \X^* \btheta_0 + \bo{m}_{[Z]} + \bo{\epsilon},
\ee
where $\bo{m}_{[Z]}=\sum_{q=1}^Q \bo{m}_{P+q}(\cdot)$.  \cite{opsomer1999root} have shown that $\wtheta$ is a consistent estimator of $\btheta_0$. Hence, under $H_0^*$, from equation (\ref{wz}) we get
\be
 \X^* \left(\X^{*T} \X^*\right)^{-1} \X^{*T} \bo{Y}^* 
= \X^* \wtheta \ \overset{\mbox{a.s.}}{=} \ \X^* \btheta_0 = \sum_{p=1}^P \bo{m}_p(\cdot) 
= \bo{m}_{[X]}.
\label{xt}
\ee
Using a similar technique of equation (\ref{sm0}) it can be shown that
\be
 \bo{S}_q^* \bo{m}_p = \bo{O}_{n,1} \mbox{ a.s. for all }  p =1,2, \cdots , P \mbox{ and } q  = 1,2, \cdots , Q.
 \label{sqp}
 \ee
 So, combining (\ref{xt}) and (\ref{sqp}) we get
 \ba
 \bo{Y}^{*T} \bo{S}_{[Z]}^{*T} \X^* \left(\X^{*T}  \X^*\right)^{-1} \X^{*T} \bo{Y}^* &=& \bo{Y}^{*T} \bo{S}_{[Z]}^* \bo{m}_{[X]} + o_p(1) \nn
 &=& \bo{m}_{[Z]}^T \bo{m}_{[X]} + o_p(1) \nn
 &=& o_p(1).
 \label{szx}
 \ea
 Hence, equation (\ref{rss122*}) simplifies to
 \be
RSS_0^* = \bo{Y}^{*T}  \left(\bo{I}_n + \bo{S}_{[Z]}^{*T} \bo{S}_{[Z]}^* - \bo{S}_{[Z]}^* - \bo{S}_{[Z]}^{*T} - \X^* \left(\X^{*T}  \X^*\right)^{-1} \X^{*T} \right) \bo{Y}^* + o_p(1).
\label{rss12*}
\ee
Now, proceeding the same way as the proof of Corollary \ref{theorem:chi},  we  get
\begin{multline}
RSS_0^* - RSS_1^* = \bo{Y}^{*T} \Big(2\bo{S}_{[X]}^* - \bo{S}_{[X]}^{*T} \bo{S}_{[X]}^* - \bo{S}_{[Z]}^{*T} \bo{S}_{[X]}^* \\
- \bo{S}_{[X]}^{*T} \bo{S}_{[Z]}^* - \X^* \left(\X^{*T}  \X^*\right)^{-1} \X^{*T} \Big) \bo{Y}^* + o_p(1).
\end{multline}
Using equations (\ref{smp}) and (\ref{szx}) we get
\be
RSS_0^* - RSS_1^* = \bo{Y}^{*T} \left(\bo{S}_{[X]}^* -  \X^* \left(\X^{*T}  \X^*\right)^{-1} \X^{*T}  \right) \bo{Y}^* + o_p(1).
\label{rsss}
\ee
Combining  equations (\ref{rss012}) and (\ref{rss018}) we get
\be
\bo{Y}^{*T}  \bo{S}_{[X]}^* \bo{Y}^* 
= \bo{m}_{[X]}^T  \bo{S}_{[X]}^*  \bo{m}_{[X]} + 2\bo{m}_{[X]}^T  \bo{S}_{[X]}^* \bo{\epsilon} + \bo{\epsilon}^T \bo{S}_{[X]}^* \bo{\epsilon} +o_p(1).
\label{mm}
\ee
Using CLT it is easy to establish that $\bo{m}_{[X]}^T  \bo{S}_{[X]}^* \bo{\epsilon}=o_p(1)$. From equation (\ref{msm2}), we have $\bo{m}_{[X]}^T  \bo{S}_{[X]}^*  \bo{m}_{[X]} = \bo{m}_{[X]}^T    \bo{m}_{[X]} + o_p(1)$. Then, equation (\ref{mm}) turns out to be
\be
\bo{Y}^{*T}  \bo{S}_{[X]}^* \bo{Y}^* 
= \bo{m}_{[X]}^T  \bo{m}_{[X]} + \bo{\epsilon}^T \bo{S}_{[X]}^* \bo{\epsilon} +o_p(1).
\label{mm1}
\ee
Note that
\begin{multline}
\bo{Y}^T \X^* \left(\X^{*T}  \X^*\right)^{-1} \X^{*T}  \bo{Y}
%
% &=& \btheta_0^T \X^{*T} \X^* \btheta_0 + \btheta_0^T \X^{*T} \bo{m}_{[Z]} + \btheta_0^T \X^{*T} \bo{\epsilon}\nn
% && \ \ +\bo{m}_{[Z]}^T \X^* \btheta_0 + \bo{m}_{[Z]}^T  \X^* \left(\X^{*T}  \X^*\right)^{-1} \X^{*T} \bo{m}_{[Z]} + \bo{m}_{[Z]}^T  \X^* \left(\X^{*T}  \X^*\right)^{-1} \X^{*T} \bo{\epsilon}\nn
% && \ \ +\bo{\epsilon}^T \X^* \btheta_0 + \bo{\epsilon}^T  \X^* \left(\X^{*T}  \X^*\right)^{-1} \X^{*T} \bo{m}_{[Z]} + \bo{\epsilon}^T  \X^* \left(\X^{*T}  \X^*\right)^{-1} \X^{*T} \bo{\epsilon}\nn
%
% = \bo{m}_{[X]}^T \bo{m}_{[X]} + \bo{m}_{[X]}^T \bo{m}_{[Z]} + \bo{m}_{[X]}^T \bo{\epsilon}\\
%  +\bo{m}_{[Z]}^T \bo{m}_{[X]} + \bo{m}_{[Z]}^T  \X^* \left(\X^{*T}  \X^*\right)^{-1} \X^{*T} \bo{m}_{[Z]} + \bo{m}_{[Z]}^T  \X^* \left(\X^{*T}  \X^*\right)^{-1} \X^{*T} \bo{\epsilon}
%  \\
%  +\bo{\epsilon}^T \bo{m}_{[X]} + \bo{\epsilon}^T  \X^* \left(\X^{*T}  \X^*\right)^{-1} \X^{*T} \bo{m}_{[Z]} + \bo{\epsilon}^T  \X^* \left(\X^{*T}  \X^*\right)^{-1} \X^{*T} \bo{\epsilon},
 %
 = \bo{m}_{[X]}^T \bo{m}_{[X]} + 2 \bo{m}_{[X]}^T \bo{m}_{[Z]} + 2 \bo{m}_{[X]}^T \bo{\epsilon}\\
 + \bo{m}_{[Z]}^T  \X^* \left(\X^{*T}  \X^*\right)^{-1} \X^{*T} \bo{m}_{[Z]} + 2 \bo{m}_{[Z]}^T  \X^* \left(\X^{*T}  \X^*\right)^{-1} \X^{*T} \bo{\epsilon}
 \\
 + \bo{\epsilon}^T  \X^* \left(\X^{*T}  \X^*\right)^{-1} \X^{*T} \bo{\epsilon}. 
 \label{xxy}
  \end{multline}
  Using condition (C8) and equation (\ref{szx}) it can be shown that the second and the forth terms in equation (\ref{xxy}) tend to zero in probability; and by CLT the third and the fifth terms is asymptotically zero.
  Therefore
  \be
\bo{Y}^T \X^* \left(\X^{*T}  \X^*\right)^{-1} \X^{*T}  \bo{Y}
  = \bo{m}_{[X]}^T \bo{m}_{[X]}  
 + \bo{\epsilon}^T  \X^* \left(\X^{*T}  \X^*\right)^{-1} \X^{*T} \bo{\epsilon} + o_p(1).
 \ee
As $\frac{1}{n}\bo{m}_{[X]}^T \bo{1}_n \overset{\mbox{a.s.}}{=}  E(\bo{m}_{[X]})=  0$, we get from the above equation
\ba
\bo{Y}^{*T} \X^* \left(\X^{*T}  \X^*\right)^{-1} \X^{*T}  \bo{Y}^*
&=& \bo{Y}^T \X^* \left(\X^{*T}  \X^*\right)^{-1} \X^{*T}  \bo{Y} \nn
&& \ \ \ - \frac{\bo{1}_n^T }{n}\bo{Y}^T \X^* \left(\X^{*T}  \X^*\right)^{-1} \X^{*T}  \bo{Y} \frac{ \bo{1}_n}{n} \nn
&=&  \bo{Y}^T \X^* \left(\X^{*T}  \X^*\right)^{-1} \X^{*T}  \bo{Y}
+ o_p(1)\nn
&=& \bo{m}_{[X]}^T \bo{m}_{[X]}  
 + \bo{\epsilon}^T  \X^* \left(\X^{*T}  \X^*\right)^{-1} \X^{*T} \bo{\epsilon} + o_p(1)\nn
\label{yxy}
\ea
% Similarly we get
% % \begin{multline}
% % \bo{Y}^{*T} \X^* \left(\X^{*T}  \X^*\right)^{-1} \X^{*T} \bo{S}_{[Z]}^* \bo{Y}^*
% % = \bo{m}_{[X]}^T \bo{m}_{[Z]}  + \bo{m}_{[Z]}^T  \X^* \left(\X^{*T}  \X^*\right)^{-1} \X^{*T} \bo{m}_{[Z]}  \\
% %   + \bo{\epsilon}^T \X^* \left(\X^{*T}  \X^*\right)^{-1} \X^{*T}  \bo{m}_{[Z]} 
% %   + o_p(1).
% %   \label{xz}
% % \end{multline}
% \begin{multline}
% \bo{Y}^{*T} \X^* \left(\X^{*T}  \X^*\right)^{-1} \X^{*T} \bo{S}_{[Z]}^* \bo{Y}^*
% = \bo{m}_{[X]}^T \bo{S}_{[Z]}^* \bo{m}_{[X]} + \bo{m}_{[X]}^T \bo{S}_{[Z]}^* \bo{m}_{[Z]} + \bo{m}_{[X]}^T \bo{S}_{[Z]}^* \bo{\epsilon} \\ + \bo{m}_{[Z]}^T  \X^* \left(\X^{*T}  \X^*\right)^{-1} \X^{*T} \bo{S}_{[Z]}^* \bo{m}_{[X]}  + \bo{m}_{[Z]}^T  \X^* \left(\X^{*T}  \X^*\right)^{-1} \X^{*T} \bo{S}_{[Z]}^* \bo{m}_{[Z]}\\
% + \bo{m}_{[Z]}^T  \X^* \left(\X^{*T}  \X^*\right)^{-1} \X^{*T} \bo{S}_{[Z]}^* \bo{\epsilon}
% + \bo{\epsilon}^T  \X^* \left(\X^{*T}  \X^*\right)^{-1} \X^{*T} \bo{S}_{[Z]}^* \bo{m}_{[X]} \\
% + \bo{\epsilon}^T  \X^* \left(\X^{*T}  \X^*\right)^{-1} \X^{*T} \bo{S}_{[Z]}^* \bo{m}_{[Z]}
%   + \bo{\epsilon}^T \X^* \left(\X^{*T}  \X^*\right)^{-1} \X^{*T}  \bo{S}_{[Z]}^* \bo{\epsilon}
%   + o_p(1).
% \end{multline}
% Therefore
% \begin{multline}
% \bo{Y}^{*T} \X^* \left(\X^{*T}  \X^*\right)^{-1} \X^{*T} \bo{S}_{[Z]}^* \bo{Y}^*
% =  \bo{m}_{[Z]}^T  \X^* \left(\X^{*T}  \X^*\right)^{-1} \X^{*T} \bo{S}_{[Z]}^* \bo{m}_{[Z]}\\
%   + \bo{\epsilon}^T \X^* \left(\X^{*T}  \X^*\right)^{-1} \X^{*T}  \bo{S}_{[Z]}^* \bo{\epsilon}
%   + o_p(1).
%   \label{xz}
% \end{multline}
Combining (\ref{mm1}) and (\ref{yxy}), we get from (\ref{rsss})
\ba
RSS_0^* - RSS_1^* &=& \bo{\epsilon}^T \bo{S}_{[X]}^* \bo{\epsilon}
- \bo{\epsilon}^T  \X^* \left(\X^{*T}  \X^*\right)^{-1} \X^{*T} \bo{\epsilon} 
+ o_p(1) \nn
&\approx& \sum_{p=1}^P \bo{\epsilon}^T \left(\bo{S}_p 
-   \bo{R}_{p,1} \left(\bo{R}_{p,1}^T  \bo{R}_{p,1}\right)^{-1} \bo{R}_{p,1}^T
\right) \bo{\epsilon} 
+ o_p(1),
\label{rss**}
\ea
where $\bo{R}_{p,1} = {}_{0}^{r_p}\X_{(p)}$ and ${}_{a}^{b}\X_{(p)}$ is defined in (\ref{xp}). 
It can be shown that 
\be
\bo{S}_p =  \bo{R}_{p,2} \left(\bo{R}_{p,2}^T  \bo{R}_{p,2} \right)^{-1} \bo{R}_{p,2}^T  ,
\label{hat_mat}
\ee
where $\bo{R}_{p,2} = {}_{0}^{k_p-1}\X_{(p)}$. 
So $\bo{S}_p$ may be regarded as the hat matrix in context of the classical regression in fitting of a $k_p$ degree polynomial. Equation (\ref{hat_mat}) shows that columns of the matrix $\bo{S}_p$ form an orthogonal basis for the column space of $\bo{R}_{p,2}$. Similarly, columns of $\bo{R}_{p,1} \left(\bo{R}_{p,1}^T  \bo{R}_{p,1}\right)^{-1} \bo{R}_{p,1}^T$ form an orthogonal basis for the column space of 
$\bo{R}_{p,1}$. Using some matrix calculations it can be shown that
\be
\bo{S}_p -   \bo{R}_{p,1} \left(\bo{R}_{p,1}^T  \bo{R}_{p,1}\right)^{-1} \bo{R}_{p,1}^T
= \bo{R}_p \left(\bo{R}_p^T  \bo{R}_p\right)^{-1} \bo{R}_p^T, \mbox{ a.s.,}
\ee
where $\bo{R}_p = {}_{r_p+1}^{k_p-1}\X_{(p)}$. Now $\bo{R}_p \left(\bo{R}_p^T  \bo{R}_p\right)^{-1} \bo{R}_p^T$ is an idempotent matrix with rank $(k_p-r_p-1)$. Hence 
\be
\frac{1}{\sigma^2} \bo{\epsilon}^T\bo{R}_p \left(\bo{R}_p^T  \bo{R}_p\right)^{-1} \bo{R}_p^T \bo{\epsilon}  
\overset{a}{\equiv}  \bo{U}_p^T \bo{U}_p 
\sim \chi^2(k_d-r_p-1),
\ee
where
\be
\bo{U}_p \overset{a}{\equiv} \frac{1}{\sigma}  \left(\bo{R}_p^T  \bo{R}_p\right)^{-1/2} \bo{R}_p^T \bo{\epsilon} .
\ee
So $(k_d-r_p-1)$ components of $\bo{U}_p$ are i.i.d. standard normal variables. 
From equation (\ref{rss**}), we get
 \be
\frac{1}{\sigma^2}(RSS_0^* - RSS_1^*) 
\overset{a}{\equiv}  \sum_{p=1}^P \bo{U}_p^T \bo{U}_p,
\label{rss1g*}
\ee
where 
\ba
{\rm Cov}(\bo{U}_p, \bo{U}_{p'}) &=& \lim_{n \rightarrow \infty}
\frac{1}{\sigma^2}  {\rm Cov} \left(\left(\bo{R}_p^T  \bo{R}_p\right)^{-1/2} \bo{R}_p^T \bo{\epsilon}, \left(\bo{R}_{p'}^T  \bo{R}_{p'}\right)^{-1/2} \bo{R}_{p'}^T \bo{\epsilon}\right) \nn
&=&
\lim_{n \rightarrow \infty}  \left(\bo{R}_p^T  \bo{R}_p\right)^{-1/2} \bo{R}_p^T  \bo{R}_{p'} \left(\bo{R}_{p'}^T  \bo{R}_{p'}\right)^{-1/2}  .
\ea
Rest of the proof is done using the same technique as the proof of Corollary \ref{theorem:chi}.
% As the rank of the matrix $\bo{S}_p^* -  \X_{(p)}^* \left(\X_{(p)}^{*T}  \X_{(p)}^*\right)^{-1} \X_{(p)}^{*T}$ is $(k_p - r_p - 1)$ for all $p=1,2, \cdots,P$. 
% The proof of the theorem is followed from the fact that, under $H_0^*$,
% \be
% \frac{1}{\sigma^2}(RSS_0^* - RSS_1) 
% \overset{a}{\sim} \chi^2\left(\sum_{p=1}^P (k_p - r_p - 1) \right).
% \ee
\end{proof}

\begin{proof}[Theorem \ref{theorem:chi3}]
In this case, we can show that
 \begin{multline}
RSS_0 - RSS_1 
\approx \sum_{p=1}^{P_1} \bo{\epsilon}^T \left(\bo{S}_p 
-   \bo{R}_{p,1} \left(\bo{R}_{p,1}^T  \bo{R}_{p,1}\right)^{-1} \bo{R}_{p,1}^T \right) \bo{\epsilon} 
+ \sum_{p=P_1+1}^{P} \bo{\epsilon}^T \bo{S}_p^*  \bo{\epsilon}
+ o_p(1),
\end{multline}
where $\bo{R}_{p,1} = {}_{0}^{r_p}\X_{(p)}$.
Hence, the proof of the theorem follows from  Corollaries \ref{theorem:chi2} and \ref{theorem:chi}. 
\end{proof}

\begin{proof}[Theorem \ref{theorem:power2}]
Combining  steps of  Theorems \ref{theorem:chi3} and \ref{theorem:power}, we get the proof of the current theorem. 
\end{proof}

\end{document}